\documentclass[journal, 10 pt]{article}  % Comment this line out

\usepackage{graphicx}

\usepackage{caption}
\usepackage{epstopdf}
\usepackage{mathtools}
\usepackage{amsthm}
\usepackage{amssymb, amsopn}
\usepackage{mathrsfs}
\usepackage{bbm}
\usepackage{euscript}
\usepackage[sans]{dsfont}
\usepackage{macros} % our file
\usepackage{url}
\usepackage{mdwmath}
\usepackage{mdwtab}
\usepackage{eqparbox}
\usepackage{float}
\usepackage{subfigure}
\usepackage{nicefrac}
\usepackage{dblfloatfix}
\usepackage{placeins}

\usepackage[all]{xy}

\usepackage[usenames,dvipsnames]{xcolor} % Custom colours - call before macros which use names.

\usepackage[normalem]{ulem}
\makeatletter
\let\NAT@parse\undefined
\makeatother
\usepackage{url}
\urldef{\mccarthy}\url{philip.mccarthy@uwaterloo.ca}
\urldef{\nielsen}\url{cnielsen@uwaterloo.ca}
\urldef{\combined}\url{{philip.mccarthy; cnielsen}@uwaterloo.ca}

\newif\ifhideproofs
\ifhideproofs
\usepackage{environ}
\NewEnviron{hide}{}

\fi
\title{\LARGE \bf Local Synchronization of Sampled-Data Systems on Lie Groups}
\date{\vspace{-5ex}}
\author{Philip James McCarthy \qquad Christopher Nielsen
  \thanks{This research is supported by the Natural Sciences
    and Engineering Research Council of Canada (NSERC). } \thanks{The
    authors are with the Dept. of Electrical and Computer
    Engineering, University of Waterloo, Waterloo ON, N2L 3G1 Canada. \combined}
}

\begin{document}

\maketitle
\thispagestyle{empty}
\pagestyle{empty}

% Show page numbers, change to {empty} before submitting
%\thispagestyle{plain}
%\pagestyle{plain}

% A B S T R A C T
% ---------------

%======================================================================
\begin{abstract}
We present a smooth distributed nonlinear control law for local synchronization of identical driftless kinematic agents on a Cartesian product of matrix Lie groups with a connected communication graph. If the agents are initialized sufficiently close to one another, then synchronization is achieved exponentially fast. We first analyze the special case of commutative Lie groups and show that in exponential coordinates, the closed-loop dynamics are linear. We characterize all equilibria of the network and, in the case of an unweighted, complete graph, characterize the settling time and conditions for deadbeat performance. Using the Baker-Campbell-Hausdorff theorem, we show that, in a neighbourhood of the identity element, all results generalize to arbitrary matrix Lie groups.
\end{abstract}
%======================================================================

%======================================================================
\section{Introduction}
%======================================================================

The sampled-data setup is ubiquitous in applied control. In the LTI case, the plant may be exactly discretized, and a discrete-time controller can be designed such that closed-loop stability is achieved for non-pathological sampling-periods. Such stability guarantees cannot generally be enforced for nonlinear plants, as nonlinear ODEs generally do not have closed-form solutions. The standard approach to nonlinear sampled-data control design, emulation, is to approximate the discretized plant dynamics. If the sampling-period is sufficiently small, then the actual closed-loop system is stable. This technique has two key shortcomings~\cite{Nesic2004}: 1) it may not be possible for a given approximate discretization method, e.g., Euler's method; 2) it relies on fast sampling, which may not be possible due to hardware limitations.

The limitations of emulation do not necessarily pose a problem for the class of systems on matrix Lie groups, which are nonlinear, yet have dynamics that yield exact closed-form solutions~\cite{Elliott2009}, thereby enabling direct design. To our knowledge, sampled-data control of systems on Lie groups has not yet been explored in the literature. However, the closely related class of bilinear systems has been studied in the discrete-time~\cite{Elliott2009} and sampled-data settings~\cite{Sontag1986}.

Many engineering systems are modelled on Lie groups. The motion of robots in a plane is modelled on $\SE{2}$~\cite{Justh2004}, and their motion in space, such as that of UAVs', is modelled on $\SE{3}$~\cite{Roza2014}. Quantum systems evolve on the unitary groups $\mathsf{U}(n)$~\cite{Lohe2009} and $\mathsf{SU}(n)$~\cite{Altafini2012,Albertini2015}. Some circuits can be modelled using Lie groups~\cite{Willsky1976} while oscillator networks~\cite{Dorfler2014} evolve on~$\SO{2}$.

Synchronization of networks on $\SE{3}$ was achieved using passivity in~\cite{Igarashi2009}. Synchronization under sampling was studied for a network of Kuramoto-like oscillators in~\cite{Giraldo2013} and harmonic oscillators with a time-varying period in~\cite{Sun2014}, and path following nonlinear agents in~\cite{Ihle2007}, but the analyses in these works were not conducted from a Lie-theoretic perspective. The Kuramoto network model was extended from $\SO{2}$ to $\SO{n}$ in~\cite{Lohe2009}. A framework for coordinated motion on Lie groups was developed in~\cite{Sarlette2010}, where the synchronization problem that we consider is a special case of what the authors call bi-invariant coordination. In~\cite{Dong2013}, linear consensus algorithms were applied to systems on Lie groups in continuous-time. The most salient difference between the current paper and~\cite{Sarlette2010,Dong2013} is the consideration of the sampled-data setup. Further, in contrast to~\cite{Dong2013}, we take a global perspective, and explore the geometry of the problem to much greater depth.

Lie groups are not vector spaces, but their structure facilitates analysis and control design in global coordinates. The Lie structure has been leveraged, for example, for motion tracking on $\SE{3}$~\cite{Park2014}, and the control of UAV~\cite{Forbes2013} and spacecraft~\cite{Egeland1994} orientation on $\SO{3}$. We too take a global perspective in our control design and, when possible, analysis. 

We present a control law that achieves synchronization for a network of identical agents on any matrix Lie group with driftless dynamics with a connected communication graph. The current paper generalizes and extends the results of~\cite{McCarthy2017}, which considered only unweighted graphs with agents on one-parameter Lie subgroups. The controller requires that each agent have access to its relative state with respect to each of its neighbours. For example, on $\SE{3}$, relative position and orientation can be attained using machine vision~\cite{Mahboubi2011}. We examine the special case where the error dynamics evolve on a Cartesian product of one-parameter subgroups -- a ``generalized cylinder'' -- and the general case. We prove that in both cases, that if the agents are initialized sufficiently close to one another, that synchronization is achieved exponentially fast. For a generalized cylinder, we characterize the performance in the case of an unweighted, complete graph.

%======================================================================
\subsection{Notation and Terminology}
%======================================================================
If $N \in \Nat$, then $\Nat_N \coloneqq \set{1, \ldots, N}$. Given a matrix $M \in \Complex^{n \times n}$, $M^\top$ is its (non-Hermitian) transpose, and $\lambda_\mathrm{max}(M)$ and $\lambda_\mathrm{min}(M)$ are its its eigenvalues of greatest and least magnitude, respectively. If $x \in\Complex^n$, then $\|x\|$ is its Euclidean norm; if $M \in \Complex^{n \times n}$, then $\|M\|$ is its induced Euclidean norm. Let $\mathbf 1_n \in \Real^n$ and $\mathbf 0_n \in \Real^n$ denote the column vector of ones and zeros, respectively. Let $\mathbf 0_{m \times n} \in \Real^{m \times n}$ denote the matrix of zeros. Let $\Real^-$ denote the set of nonpositive real numbers. Given an equivalence relation $\sim$ on a set $R$, and an element $x \in R$, let $[x] \in R/\sim$ be the coset containing $x$.

Weighted, directed graphs are used to model communication constraints between agents. A graph $\mathcal G$ is a triple $(\mathcal{V}, \mathcal{E},w)$ consisting of a finite set of vertices $\mathcal{V} = \Nat_N$, a set of edges $\mathcal E \subseteq \mathcal{V} \times \mathcal{V}$, and a weight function $w : \mathcal{E} \to [0, 1] \subset \Real$. The weight $w_{ij} \coloneqq w((i,j))$ is nonzero only if $(i, j) \in \mathcal{E}$. If agent $i$ has access to its relative state with respect to agent $j$, then $(i,j) \in \mathcal{E}$. Define vertex $i$'s neighbour set as $\mathcal N_i \coloneqq \{j \in \Nat_N: (i,j) \in \mathcal{E}\}$. We assume that $\mathcal G$ has no self-loops. If $\mathcal{G}$ is unweighted, then for all $i \neq j \in \Nat_N$, $w_{ij} \in \{0,1\}$. Associated with $\mathcal G$ is the Laplacian $L \in \Real^{N \times N}$, defined elementwise as
\begin{equation*}
	L_{ij} = 
		\left\{\begin{array}{rl}
			-w_{ij}, & i \neq j, \\
			\sum_{j \in \mathcal N_i}w_{ij}, & i = j.
		\end{array}\right.
\end{equation*}
The $i$th row of the Laplacian $L$ is denoted $\ell_i$.

%======================================================================
\section{Sampled-Data Synchronization Problem}
%======================================================================

We consider a network of $N$ controlled agents, each modelled by the differential equation
\begin{equation}\label{eq:ctplant}
  \dot X_i = X_i\left(\sum_{j=1}^mB_{i,j}u_{i,j}\right), \qquad i \in \Nat_N.
\end{equation}
Here $X_i \in \G$ where $\G \subset \GL{n,\Complex}$ is an $m$-dimensional connected matrix Lie group over the complex field $\Complex$ which includes, as a special case, real matrix Lie groups. The matrices $B_{i,j}$ are elements of the Lie algebra $\g$, which is a vector space over a field $\Field$ equal to either $\Complex$ or $\Real$, associated with $\G$, and $u_i \coloneqq (u_{i,1}, \ldots, u_{i,m}) \in \mathbb F^m$ is the control input. Note that the Lie algebra of a complex Lie group $\G$ may in fact be a real vector space. For example, the Lie algebra $\mathfrak{su}(2)$ of the complex Lie group $\SU{2}$ is a vector space over the field of reals despite its vectors being matrices with possibly complex entries. Equation~\eqref{eq:ctplant} is a kinematic model of a system evolving on a matrix Lie group $\G$. Each agent is assumed to be fully actuated in the sense that
\begin{equation*}
  \left(\forall i \in \Nat_N\right) \ \Span_{\mathbb F}\set{B_{i,1}, \ldots, B_{i,m}} = \g.
	%\left(\forall i \in \Nat_N\right) \ \Lie(\{B_{i,1}, \ldots, B_{i,n}\}) = \g,
\end{equation*}
Under this assumption, without loss of generality, we take the
system~\eqref{eq:ctplant} to be driftless since the inputs $u_i$,
$i \in \Nat_N$, can be chosen to cancel any drift vector field. We are
interested in the sampled-data control of this multi-agent system in
which each agent's control law is implemented on an embedded computer,
which we explicitly model using the setup in Figure~\ref{fig:SD}. The
blocks $H$ and $S$ in Figure~\ref{fig:SD} are, respectively, the ideal
hold and sample operators. Sample and hold are, respectively,
idealized models of A/D and D/A conversion.
 %%%%%%%%%%%%%%%
\begin{figure}[h!]
	\centering
	\includegraphics[width=\textwidth]{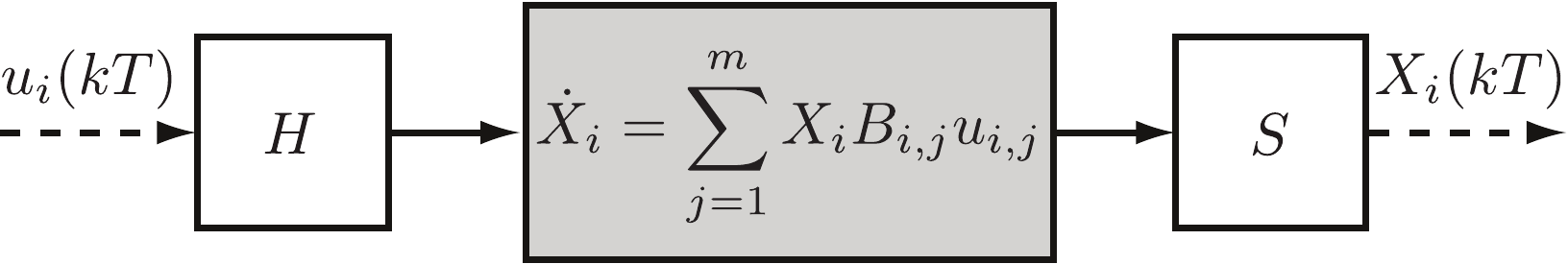}
	\caption{Sampled-data agent on a matrix Lie group $\G$.}
	\label{fig:SD}
\end{figure}
%%%%%%%%%%%%%%%
The following assumption is made throughout this paper.
\begin{assumption}\label{ass:HS}
All sample and hold blocks operate at the same period $T > 0$ and the blocks are synchronized for the multi-agent system~\eqref{eq:ctplant}. \hfill $\blacktriangleleft$
\end{assumption}

Under Assumption~\ref{ass:HS}, letting $X_i[k] \coloneqq X_i(kT)$ and $u_i[k] \coloneqq u_i(kT)$, the discretized dynamics of each agent are given by
\begin{equation}\label{eq:dtplant}
  X_i^+ = X_i\exp{\left(T\sum_{j=1}^mB_{i,j}u_{i,j}\right)}, \qquad i \in \Nat_N
\end{equation}
which is an exact discretization of~\eqref{eq:ctplant}. For each $i \in \Nat_N$, define $\Omega_ i \coloneqq \sum_{j=1}^mB_{i,j}u_{i,j} \in \g$. Then the discrete-time dynamics can be compactly expressed as
\begin{equation}\label{eq:dtplant2}
	X_i^+ = X_i\exp{\left(T\Omega_i\right)}, \qquad i \in \Nat_N.
\end{equation}
%
%The state space of the multi-agent system is the product manifold $\mathsf{M} \coloneqq \G \times \G \times \cdots \times \G$ ($N$-times).

%----------------------------------------------------------------------
\subsection{The Synchronization Problem}
%----------------------------------------------------------------------

Given a network of $N$ agents with kinematic dynamics~\eqref{eq:dtplant2}, we define the error quantities $E_{ij} \coloneqq X_i^{-1}X_j$, $i,j \in \Nat_N$. Observe that $E_{ij} = I$ if, and only if, $X_i = X_j$. The error matrix $E_{ij}$ is called left-invariant~\cite{LagTruMah10}, since for all $X \in \G$, $(XX_i)^{-1}(XX_j) = X_i^{-1}X_j$. The class of systems considered has $\G$ for its state space, which is generally not a vector space, so we do not use $X_i - X_j$ as a measure of error.

\begin{myproblem}{Local Synchronization on Matrix Lie Groups}\label{prob:KN}
Given a network of $N$ agents with continuous-time dynamics~\eqref{eq:ctplant}, sampling period $T > 0$ and an unweighted, connected communication graph $\mathcal G = (\mathcal{V}, \mathcal{E})$, find, if possible, distributed control laws $\Omega_i$, $i \in \Nat_N$, such that for all initial errors in a neighbourhood of the identity in $\G^N$, for all $i,j \in \Nat_N$, $E_{ij} \to I$ as $t \rightarrow \infty$.
\end{myproblem}

By a distributed control law we mean that for each agent $i$, the control signal $\Omega_i$ can depend on $E_{ij}$ only if $(i, j) \in \mathcal{E}$. In this paper we propose the distributed feedback control law
\begin{equation}\label{eq:Controller}
	\Omega_i \coloneqq \frac{1}{T}\log\left(\left(\prod_{j \in \mathcal N_i}E_{ij}^{\w_{ij}}\right)^{\frac{1}{K}}\right)
\end{equation}
where $K \in \Real$ is a gain and the matrix logarithm need not be the
principal logarithm. The control law~\eqref{eq:Controller} does not
require agent $i$ to know agent $j$'s state $X_j$, nor its own state
$X_i$, but instead requires knowledge of the relative state
$E_{ij}$. The expression~\eqref{eq:Controller} is well-defined so long
as the product $\prod_{j \in \mathcal N_i}E_{ij}$ has no eigenvalues
in $\Real^-$, as discussed in Section~\ref{sec:Preliminaries}. This
control law is expressed in global coordinates but we only prove local
exponential stability of the synchronized state. When the control
law~\eqref{eq:Controller} is well-defined, the closed-loop
discrete-time dynamics are
\begin{equation}\label{eq:dtplant3}
	X_i^+ = X_i\left(\prod_{j \in \mathcal N_i}E_{ij}^{\w_{ij}}\right)^{\frac{1}{K}}, \qquad i \in \Nat_N
\end{equation}
and the synchronization error dynamics are
%
%\begin{equation}
\begin{align}
	E_{ij}^+ &= \left(X_i^+\right)^{-1}X_j^+ \nonumber \\
           &= \left(\prod_{p \in \mathcal N_i}E_{ip}^{\w_{ip}}\right)^{-\frac{1}{K}}X_i^{-1}X_j\left(\prod_{q \in \mathcal N_j}E_{jq}^{\w_{jq}}\right)^{\frac{1}{K}} \nonumber\\
           &= \left(\prod_{p \in \mathcal N_i}E_{ip}^{\w_{ip}}\right)^{-\frac{1}{K}}E_{ij}\left(\prod_{q \in \mathcal N_j}E_{jq}^{\w_{jq}}\right)^{\frac{1}{K}} \label{eq:error}.
\end{align}
%\end{equation}

\begin{remark}\label{rem:Order}
The order of multiplication in~\eqref{eq:Controller} need not be common to all agents or even constant.
\hfill $\blacktriangle$
\end{remark}

%Since we consider a kinematic model, the system comes to rest if and only if all the inputs~\eqref{eq:Controller} are zero. That is, the equilibria are characterized by
%%
%\begin{equation*}
%	\left\{X \in \mathsf{M} : \forall i \in \Nat_N, \prod_{j \in \mathcal N_i}E_{ij} = I\right\}
%\end{equation*}

The control law~\eqref{eq:Controller} is motivated by exponential coordinates for Lie groups, classical consensus algorithms in $\Real^n$, and the notion of Riemannian mean of rotations on $\SO{3}$, which on a one-parameter subgroup thereof can be explicitly computed as $\prod_{i=1}^NR_i^{\frac{1}{N}}$~\cite{Moakher2002}.
 
A key advantage of direct design over emulation, is that stability can be guaranteed at the sampling instants. As mentioned in the Introduction, on $\SE{3}$, the relative error $E_{ij}$ can be computed using machine vision, where the speed of sampling is limited by the frame rate of the camera, for example, $25$ Hz~\cite{Mahboubi2011}. This limits the feasibility of emulation-based design. However, direct design does not guarantee good performance between sampling instants. But in the specific case of the plant and problem discussed in this paper, achieving synchronization at the sampling instants implies synchronization between the sampling instants.

\begin{proposition}\label{prop:Stability}
If $E_{ij}[k] = X_i[k]^{-1}X_j[k]$ asymptotically approaches $I$ as $k \rightarrow \infty$, then $E_{ij}(t) = X_i(t)^{-1} X_j(t)$ asymptotically approaches $I$ as $t \rightarrow \infty$, where $X_i(t)$ and $X_j(t)$ evolve according to~\eqref{eq:ctplant}.
\end{proposition}

\begin{proof}
If $E_{ij}[k] \to I$, then the proposed control law~\eqref{eq:Controller} satisfies $\Omega_i[k] \to \mathbf 0_{n \times n}$. Let $0 < \delta < T$. Then
\begin{equation*}
\begin{aligned}
%	X_i(kT + \delta) &= X_i(kT)\exp{\left(\delta\sum_{j=1}^nB_{i,j}u_{i,j}\right)} \\
  &\lim_{k \to \infty}E_{ij}(kT + \delta) = \lim_{k \to \infty}\exp{\left(\delta\Omega_i[k]\right)}^{-1}E_{ij}[k]\exp{\left(\delta\Omega_j[k]\right)} \\
  &\qquad= \lim_{k \to \infty}\exp{\left(\delta\Omega_i[k]\right)}^{-1}\lim_{k \to \infty}E_{ij}[k]\lim_{k \to \infty}\exp{\left(\delta\Omega_j[k]\right)} \\
  &\qquad= I^3 = I.
\end{aligned}
\end{equation*}
Since $\delta$ is arbitrary, this implies that $E_{ij}(t) \to I$.
\end{proof}

Proposition~\ref{prop:Stability} means that asymptotically stabilizing the set where $E_{ij} = I$, for all $ i,j \in \Nat_N$, at the sample instances is sufficient for solving the synchronization problem. Thus, we can conduct all analysis in the discrete-time setting and do not rely on $T$ being sufficiently small.

Our main result is the following theorem, which we prove in Section~\ref{sec:General}.

\begin{theorem}\label{thm:GeneralGain}
For any Lie group $\G$ with connected communication graph $\mathcal G$, if the gain $K$ of each agent's controller~\eqref{eq:Controller} satisfies~\eqref{eq:K}, then the equilibrium $\{E_{ij} = I : i,j \in \Nat_N\}$ is locally uniformly exponentially stable.
\end{theorem}

%======================================================================
\section{Preliminaries}
\label{sec:Preliminaries}
%======================================================================

%----------------------------------------------------------------------
\subsection{Functions of matrices}\label{sec:matfunc}
%----------------------------------------------------------------------
For every nonsingular matrix $X \in \Complex^{n \times n}$ there are (infinitely many) $A \in \Complex^{n \times n}$ such that $\exp{(A)} = X$, see~\cite[Theorem 1.27]{Higham2008}. Every such matrix $A$ is a non-primary logarithm of $X$, which we denote by $\log{(X)}$. If, in addition to being nonsingular, the matrix $X$ has no eigenvalues in $\Real^-$, then it has a (unique) principal logarithm.

\begin{theorem}[{\cite[Theorem 1.31]{Higham2008}}]\label{thm:higham}
Let $X \in \Complex^{n \times n}$ have no eigenvalues in $\Real^-$. There is a unique logarithm $A \in \Complex^{n \times n}$ of $X$, all of whose eigenvalues lie in the strip $\{z : -\pi < \mathrm{Im}(z) < \pi\}$. If $X \in \Real^{n \times n}$, then $A \in \Real^{n \times n}$.
\end{theorem}

The unique matrix $A$ from Theorem~\ref{thm:higham} is called the \textbf{principal logarithm of $X$} and is denoted $\Log(X)$. Unlike complex numbers, it is not possible to express $\log{(X)}$ as a function of $\Log{(X)}$ for arbitrary non-singular matrices. If $\|X - I\| < 1$, then
\begin{equation}\label{eq:Log}
	\Log(X) = \sum_{k=1}^\infty \frac{(-1)^{k-1}}{k} (X - I)^k.
\end{equation}

Any matrix logarithm is a right inverse of the matrix exponential, but not necessarily a left inverse. On a matrix Lie group, the principal logarithm $\Log$ is a left inverse, but only in a neighbourhood of the identity. Choose $r > 0$ such that~\eqref{eq:Log} converges on $\set{X \in \G : X = \exp{(A)}, \|A\| < r}$, e.g., $r = \Log(2)$ is a valid choice with any Lie group. Larger values of $r$ may be possible for specific Lie groups. The set
\begin{equation*}
	U \coloneqq \set{X \in \G: X = \exp{(A)}, A \in \g, \|A\| < r}
\end{equation*}
is an open neighbourhood of $I$ in $\G$ in the group topology in which $\Log : U \to \g$ provides an inverse.

Borrowing from the definition of complex powers of scalars~\cite[Chapter III, \S 6]{Lang1999} and the form of the square root of a matrix on a Lie group~\cite[Lemma 2.14]{Hall2015}, we define the $K$th root of a matrix in the following way.

\begin{definition}\label{def:KthRoot}
Let $X \in \Complex^{n \times n}$ have no eigenvalues in $\Real^-$. Given $K \in \Real$, the \textbf{principal $K$th root} of $X$ is
\begin{equation}\label{eq:Root}
	X^\frac{1}{K} \coloneqq \exp\left(\frac{1}{K}\Log(X)\right).
\end{equation}
\end{definition}

If $X \in \G$, then $X^{1/K} \in \G$, due to the Lie correspondence $\Log : \G \to \g$.

\begin{remark}
If $X^{1/K}$ is well-defined, then for $K \in \Nat$,
\begin{equation*}
	\left(X^\frac{1}{K}\right)^K = \exp\left(\sum_{i = 1}^K\frac{1}{K}\Log(X)\right) = \exp\left(\Log(X)\right) = A.
\end{equation*}
Thus, in this case, $X = X^{1/K}X^{1/K}\cdots X^{1/K}$ ($K$ times), which is the intuitive notion of a $K$th root. The somewhat indirect definition~\eqref{eq:Root} allows for $K$th roots for $K \in \Real$.
\hfill $\blacktriangle$
\end{remark}

Throughout this paper, we use an important algebraic property of the logarithm of a matrix power.
\begin{theorem}[{\cite[Theorem 11.2]{Higham2008}}]
  If $X \in \Complex^{n \times n}$ has no eigenvalues in $\Real^-$,
  then for $\alpha \in [-1,1]$, we have
  $\Log(X^\alpha) = \alpha\Log(X)$.
\end{theorem}

%----------------------------------------------------------------------
\subsection{Exponential Coordinates and One-Parameter Subgroups}
%----------------------------------------------------------------------
\label{sec:one}
\begin{definition}[One-Parameter Subgroup]
\label{def:one}
Given a Lie group $\G$, a one-parameter subgroup is a continuous morphism of groups $\phi : \Real \to \G$.
\end{definition}

Although this terminology is standard, it is technically the image of the map $\phi$ that is a subgroup of $\G$. The subgroup $\phi(\Real) \subset \G$ is a one-dimensional manifold and there exists a unique $H \in \g$ such that $\phi(\theta) = \exp(\theta H)$ for all $\theta \in \Real$~\cite[Theorem 2.13]{Hall2015}.

To generalize the concept of one-parameter subgroups to higher dimensional manifolds, we consider \textbf{generalized cylinders}. A generalized cylinder is an $m$-dimensional manifold that is diffeomorphic to $\mathbb{T}^k \times \Real^{m - k}$. Such a diffeomorphism exists if and only if there exist $m$ commutative and everywhere-linearly-independent vector fields on the manifold~\cite[pg. 274]{Arnold1989}. If the manifold is a Lie group $\G$, then this simplifies to its Lie algebra $\g$ having a commutative basis.

Let $\G$ be such a manifold and fix a commutative basis $\mathcal{H} \coloneqq \{H_1,\ldots,H_m\}$ for its Lie algebra $\g$. Consider the one-parameter groups $\phi_i : \Real \to \G$ associated with each $H_i$. The image of $\phi(t_1,\ldots,t_m) \coloneqq \phi_1(t_1)\phi_2(t_2)\cdots\phi_m(t_m)$ is $\G$. Without loss of generality, let $\phi_i$, $i \in \Nat_k$, $0 \leq k \leq m$ have nonzero kernel, and let $\phi_i$, $i \in \{k + 1,\ldots,m\}$ have zero kernel.

Fixing such a basis $\mathcal{H}$, the $\Log$ map induces local coordinates on $\G \cap U$. Given $X \in \G$, by commutativity of $H_1,\ldots,H_m$,
\begin{equation*}
\begin{aligned}
	X &= \exp(t_1 H_1)\cdots\exp(t_m H_m) \\
	&= \exp(t_1 H_1 + \cdots + t_m H_m).
\end{aligned}
\end{equation*}
If $X \in \G \cap U$, then $\Log(X) = t_1 H_1 + \cdots + t_m
H_m$.
Then, by linear independence of $H_1,\ldots,H_m$, $t_1,\ldots,t_m$ can
be uniquely determined, yielding local coordinates
$(t_1,\ldots,t_m) \in \Real^m$. Thus, a Lie group $\G$ can be locally
identified with an open subset of the vector space $\Real^m$
containing the origin. Note that, by commutativity of $\G$, these
local coordinates coincide with the familiar exponential coordinates
of both the first and second kind.

\subsection{Properties of the composed flow}
The map $\phi : \Real^m \to \G$, defined in the previous section, is critical to our analysis throughout this paper. In this section, we establish important properties of $\phi$ when $\G$ is a generalized cylinder, we then show that these properties hold approximately for any Lie group in a neighbourhood of the identity.

By definition, $\phi$ is surjective onto its image, but it is not necessarily injective. Let $p : \Real^m \to \Real^m/\Ker(\phi)$ be the projection of $\Real^m$ onto the quotient space $\Real^m/\Ker(\phi)$. There exists a unique isomorphism of groups $\phi'$ such that the following diagram commutes~\cite[Theorem 26, Corollary 1]{MacLane1999}.
\begin{equation*}
	\xymatrix{\Real^m \ar[drr]_\phi \ar[rr]^p 	&& \Real^m/\Ker(\phi) \ar[d]^{\phi'} \\
							&& \G}
\end{equation*}

The bijection $\phi'$ yields alternative global coordinates on the quotient group $\Real^m/\Ker(\phi)$; it will be used in our characterization of equilibria. When $\G$ is a generalized cylinder, the map $\phi$ has several important properties.

\begin{proposition}\label{prop:Morphism}
If $\G$ is a generalized cylinder, then $\phi : \Real^m \to \G$ is a morphism of groups.
\end{proposition}

\begin{proof}
It is clear that $\phi(\mathbf 0_m) = \exp(0)\cdots\exp(0) = I$.

Let $\p_i = (\p_i^{(1)},\ldots,\p_i^{(m)}) \in \Real^m$ and $\p_j = (\p_j^{(1)},\ldots,\p_j^{(m)}) \in \Real^m$, where $\phi(\p_i) = X_i$ and $\phi(\p_j) = X_j$. By commutativity of $H_1,\ldots,H_m$,
\begin{equation*}
\begin{aligned}
	\phi(\p_i + \p_j) &= \exp\left((t_i^{(1)} + t_j^{(1)})H_1\right)\cdots\exp\left((t_i^{(m)} + t_j^{(m)})H_m\right) \\
	&= \exp(t_i^{(1)}H_1)\cdots\exp(t_i^{(m)}H_m)\exp(t_j^{(1)}H_1)\cdots\exp(t_j^{(m)}H_m) \\
	&= \phi(\p_i)\phi(\p_j).
\end{aligned}
\end{equation*}
\end{proof}

\begin{lemma}\label{lem:LocalRoot}
Let $\G$ be a generalized cylinder. If $K > 0$ and $\phi(\p) \in U$, then $\phi(\p/K) = \phi(\p)^{1/K}$.
\end{lemma}

\begin{proof}
By the commutativity of $H_1,\ldots,H_m$,
\begin{equation*}
\begin{aligned}
	\phi\left(\frac{\p}{K}\right) &= \exp\left(\frac{t_1}{K}H_1\right)\cdots\exp\left(\frac{t_m}{K}H_m\right) \\
	&= \exp\left(\frac{1}{K}(t_1H_1 + \cdots + t_mH_m)\right) \\
	&= \exp\left(\frac{1}{K}\Log(\phi(\p))\right) \\
	&= \phi(\p)^\frac{1}{K}.
\end{aligned}
\end{equation*}
\end{proof}

%======================================================================
\section{Equilibria on Generalized Cylinders}
%======================================================================
Since we consider driftless kinematic models, the system is at equilibrium if, and only if, every agent's input is zero, i.e., for all $i \in \Nat_N$, $\Omega_i = \mathbf 0_{n \times n}$. We show that all equilibria are isolated and exhibit the same stability properties.

Hereinafter, we use the notation $\p_{ij} \coloneqq \phi^{-1}(E_{ij})$ and $\bigp \coloneqq \begin{bmatrix}\p_{11}^\top & \p_{12}^\top & \cdots & \p_{1N}^\top\end{bmatrix}^\top \in \Real^{Nm}$. Define $\bigp_\mathbb{T}$ and $\bigp_\Real$ as the projections under $p$ of $\bigp$ onto $\mathbb T^{Nk}$ and $\Real^{N(m - k)}$, respectively.

\begin{proposition}\label{prop:OPEquilibria}
If the controller~\eqref{eq:Controller} is well-defined, then the equilibria of~\eqref{eq:error} on a generalized cylinder are characterized by $\left[\frac{1}{K}(L \otimes I_k)\bigp_\mathbb{T}\right] = [\mathbf 0_{Nk}]$, $(L \otimes I_{m-k})\bigp_\Real = \mathbf 0_{N(m-k)}$, $[\p_{11}] = [\mathbf 0_m]$.
\end{proposition}

\begin{proof}
The sampled dynamics of each agent are given by~\eqref{eq:dtplant3}. Therefore, the system is at equilibrium if, and only if, for all $i \in \Nat_N$, $\exp(T\Omega_i) = I$. If $\G$ is a generalized cylinder, then, by commutativity and Definition~\ref{def:KthRoot}, this condition becomes
\begin{equation}\label{eq:EquilibriumInput}
	I = \exp(T\Omega_i) = \left(\prod_{j \in \mathcal N_i}E_{ij}^{\w_{ij}}\right)^{\frac{1}{K}} = \prod_{j \in \mathcal N_i}E_{1i}^{-\frac{\w_{ij}}{K}}E_{1j}^{\frac{\w_{ij}}{K}}.
\end{equation}
In the global coordinates admitted by $\phi'$ we have
\begin{equation*}
\begin{aligned}
	\phi'^{-1}(I) &= \sum_{j \in \mathcal N_i}\left(\phi'^{-1}\left(E_{1i}^{-\frac{\w_{ij}}{K}}\right) + \phi'^{-1}\left(E_{1j}^{\frac{\w_{ij}}{K}}\right)\right) \\
	[\mathbf 0_m] &= \sum_{j \in \mathcal N_i}\left(\left[\frac{\w_{ij}}{K}\p_{1j}\right] - \left[\frac{\w_{ij}}{K}\p_{1i}\right]\right)
	= -\left[\frac{1}{K}(\ell_i \otimes I_m)\bigp\right].
\end{aligned}
\end{equation*}
%
%where $\bigp \coloneqq \left(\phi^{-1}(E_{11}), \ldots, \phi^{-1}(E_{1N})\right)$.

We ``stack" the inputs for all agents $i$, yielding the equation
\begin{equation*}
	\left[\frac{1}{K}(L \otimes I_m)\bigp\right] = [\mathbf 0_{Nm}],
\end{equation*}
which can be rewritten as
\begin{equation*}
	\left[\frac{1}{K}\begin{bmatrix}L \otimes I_k & \mathbf 0_{Nk \times N(m - k)} \\
						\mathbf 0_{N(m - k) \times Nk} & L \otimes I_{m - k}\end{bmatrix}
						\begin{bmatrix}\bigp_\mathbb{T} \\ \bigp_\Real\end{bmatrix}\right] = 
						\left[\begin{bmatrix}\mathbf 0_{Nk} \\ \mathbf 0_{N(m - k)}\end{bmatrix}\right].
\end{equation*}
By assumption, $\Ker(\phi_i) = \{0\}$ for $i \in \{k+1,\ldots,m\}$, so the condition $\left[\frac{1}{K}(L \otimes I_{m - k})\right] = [\mathbf 0_{N(m - k)}]$ simplifies to equality on $\Real^{N(m - k)}$, rather than congruence on a quotient space. Lastly, since $\p_{11}$ is the error of agent $1$ with itself, $[\p_{11}] = [\mathbf 0_m]$.
\end{proof}

\begin{proposition}
On a generalized cylinder, all equilibria are isolated.
\end{proposition}

\begin{proof}
By assumption, $\Ker(\phi_i) = \{0\}$, for all $i \in \{k+1,\ldots,m\}$. Thus $\Real/\Ker(\phi_i) \cong \Real$. Thus, for $j \in \Nat_N$, $[\p_{1j}^{(i)}] = [0]$ simplifies to $\p_{1j}^{(i)} = 0$.

By assumption, $\Ker(\phi_i) \neq \{0\}$, for all $i \in \Nat_k$. The
map $\phi_i$ can be viewed as a flow, thus, by~\cite[Theorem
2.12]{Bhatia1970}, there exists a $d_i > 0$ such that for every
$r_i \in [0]$, we have $r_i = q_id_i$ for some $q_i \in \Int$. Thus,
for all $j \in \Nat_N$, if $\p_{1j}^{(i)}, \bar\p_{1j}^{(i)} \in [0]$,
$\p_{1j}^{(i)} \neq \bar\p_{1j}^{(i)}$, then
$|\p_{1j}^{(i)} - \bar\p_{1j}^{(i)}| \geq d_i$.
\end{proof}

\begin{proposition}\label{prop:1PTorus}
On a generalized cylinder, every equilibrium has the same stability properties as the identity.
\end{proposition}

\begin{proof}
Let $\{\Xi_{i1},\ldots,\Xi_{iN}\} \in \G^N$ be an equilibrium. Define $\bar E_{ij} \coloneqq \Xi_{ij}^{-1}E_{ij}$. Then $\bar E_{ij} = I$ if and only if $\bar E_{ij} = \Xi_{ij}$. The error dynamics~\eqref{eq:error} can be expressed in terms of $\bar E_{ij}$:
\begin{equation*}
\begin{aligned}
  \Xi_{ij}\bar E_{ij}^+ &= \left(\prod_{p \in \mathcal N_i}(\Xi_{1i}\bar E_{1i})^{-\w_{ip}}(\Xi_{1p}\bar E_{1p})^{\w_{ip}}\right)^{-\frac{1}{K}}\Xi_{ij}\bar E_{ij} 
   \times \left(\prod_{q \in \mathcal N_j}(\Xi_{1j}\bar E_{1j})^{-\w_{jq}}(\Xi_{1q}\bar E_{1q})^{\w_{jq}}\right)^{\frac{1}{K}} \\
  \bar E_{ij}^+ &= \left(\prod_{p \in \mathcal N_i}\bar E_{1i}^{-\w_{ip}}\bar E_{1p}^{\w_{ip}}\right)^{-\frac{1}{K}}\bar E_{ij}\left(\prod_{q \in \mathcal N_j}\bar E_{1j}^{-\w_{jq}}\bar E_{1q}^{\w_{jq}}\right)^{\frac{1}{K}} 
   \left(\prod_{p \in \mathcal N_i}\Xi_{1i}^{-\w_{ip}}\Xi_{1p}^{\w_{ip}}\right)^{-\frac{1}{K}}\left(\prod_{q \in \mathcal N_j}\Xi_{1j}^{-\w_{jq}}\Xi_{1q}^{\w_{jq}}\right)^{\frac{1}{K}} \\
  &= \left(\prod_{p \in \mathcal N_i}\bar E_{ip}^{\w_{ip}}\right)^{-\frac{1}{K}}\bar E_{ij}\left(\prod_{q \in \mathcal N_j}\bar E_{jq}^{\w_{jq}}\right)^{\frac{1}{K}} 
   \left(\prod_{p \in \mathcal N_i}\Xi_{ip}^{\w_{ip}}\right)^{-\frac{1}{K}}\left(\prod_{q \in \mathcal N_j}\Xi_{jq}^{\w_{jq}}\right)^{\frac{1}{K}}.
\end{aligned}
\end{equation*}

By~\eqref{eq:EquilibriumInput}, the $\Xi_{ij}$ product terms equal
identity, thus the $\bar E_{ij}$ dynamics have the same form as the
dynamics~\eqref{eq:error} of $E_{ij}$, and therefore have the same
qualitative behaviour.
\end{proof}

Proposition~\ref{prop:1PTorus} says that the dynamics near every equilibrium ``look the same''. Thus, by analyzing only the equilibrium at identity, we characterize the behaviour near all equilibria.

%======================================================================
\section{Synchronization on Generalized Cylinders}
%======================================================================

In this section, we consider the case where $\G$ is a generalized
cylinder. This means that $X_i \in \G$, $i \in \Nat_N$, which implies
$E_{ij} \in \G$, for all $i,j \in \Nat_N$. This is done only to
simplify discussion. The results of this section hold under the weaker
assumption that the errors $E_{ij}$ lie on a generalized cylinder.

\begin{proposition}\label{prop:invar}
Any generalized cylinder $\G$, on which~\eqref{eq:Controller} is well-defined for all forward time, is positively invariant for~\eqref{eq:error}.
\end{proposition}

\begin{proof}
Let $k \in \Int$ be arbitrary and suppose that for all $i, j \in \Nat_N$, $E_{ij}[k] \in \G$. Let $q \in \Nat_N$ be arbitrary. By elementary group theory $\prod_{j \in \mathcal{N}_q} E_{qj}[k] \in \G$ and therefore, by hypothesis, its $K$th root is well-defined. Thus, by definition of $\log$, $\exp{(T\Omega_q[k])} \in \G$. Since $E_{ij}^+ = \exp(-T\Omega_i)E_{ij}\exp(T\Omega_j)$,  we have $E_{ij}[k+1] \in \G$. Induction on the time index proves positive invariance of $\G$.
\end{proof}

%\begin{proposition}
%Any generalized cylinder $\G$ is positively invariant for the continuous-time dynamics~\eqref{eq:ctplant} under the discrete-time control law~\eqref{eq:Controller}.
%\end{proposition}
%
%\begin{proof}
%Let $\G$ be a generalized cylinder generated by $\mathcal{H}$. Suppose $E_{ij}[0] = E_{ij}(0) \in \G$ for $i,j \in \Nat_N$. By Proposition~\ref{prop:invar}, $E_{ij} \in \G$ for all $k \geq 0$. Let $0 < \delta < T$. Then
%%
%\begin{equation*}
%	E_{ij}(kT + \delta) = \exp{\left(-\delta\Omega_i[k]\right)}E_{ij}[k]\exp{\left(\delta\Omega_j[k]\right)}.
%\end{equation*}
%%
%But $\Omega_i[k] = \p H$ for some $\p \in \Real^m$, $H \in \Lie\{H_1,\ldots,H_m\}$ and therefore, since $\delta$ is real, the exponentiated term is an element of $\G$. Since $E_{ij}[k] \in \G$ we conclude that $E_{ij}(kT + \delta) \in \G$.
%\end{proof}

Using the exponential coordinates from Section~\ref{sec:one} we identify each relative error $E_{ij} \in \G \cap U$ with its exponential coordinates $\p_{ij} \in \Real^m$. We henceforth impose that the synchronization errors and their products over neighbour sets be close to the identity.

\begin{assumption}\label{ass:Close}
For all $i,j \in \Nat_N$, we have $E_{ij} \in U$ and $\prod_{j \in \mathcal N_i}E_{ij}^{w_{ij}} \in U$. \hfill $\blacktriangleleft$
\end{assumption}

Let $\G$ be a generalized cylinder and suppose $E_{ij}[0] \in \G$, $i,j \in \Nat_N$.  By Proposition~\ref{prop:invar}, for all $k \geq 0$, $E_{ij}[k] \in \G$.

It follows from its definition that $\phi$ is a local diffeomorphism in a neighbourhood of the identity element. We first apply the identity that for all $i,j \in \Nat_N$, $E_{ij} = E_{1i}^{-1}E_{1j}$ to~\eqref{eq:error}:
\begin{equation}\label{eq:ErrorAll1}
	E_{ij}^+ = \left(\prod_{p \in \mathcal N_i}E_{1i}^{-w_{ip}}E_{1p}^{w_{ip}}\right)^{-\frac{1}{K}}E_{ij}\left(\prod_{q \in \mathcal N_j}E_{1j}^{-w_{jq}}E_{1q}^{w_{jq}}\right)^{\frac{1}{K}}.
\end{equation}
Applying Proposition~\ref{prop:Morphism} and Lemma~\ref{lem:LocalRoot} to~\eqref{eq:ErrorAll1}, we have
%%
%\begin{equation}\label{eq:AngleDynamics1}
%	\p_{ij}^+ = \left(-\frac{1}{K}\sum_{p \in \mathcal N_i}\p_{ip}\right) + \p_{ij} + \left(\frac{1}{K}\sum_{q \in \mathcal N_j}\p_{jq}\right).
%\end{equation}
%
%Under Assumption~\ref{ass:Close}, equation~\eqref{eq:AngleDynamics1} can be rewritten as
%
\begin{equation*}\label{eq:AngleDynamics2}
\begin{aligned}
	\p_{ij}^+ &= \p_{ij} - \frac{1}{K}\sum_{p \in \mathcal N_i}\w_{ip}(\p_{1p} - \p_{1i}) + \frac{1}{K}\sum_{q \in \mathcal N_j}\w_{jq}(\p_{1q} - \p_{1j}) \\
	&= \p_{ij} - \frac{1}{K}\Bigg(\sum_{p \in \mathcal N_i}\w_{ip}\p_{1p} - \Bigg(\sum_{p \in \mathcal N_i}\w_{ip}\Bigg)\p_{1i}\Bigg) 
		+ \frac{1}{K}\Bigg(\sum_{q \in \mathcal N_j}\w_{jq}\p_{1q} - \Bigg(\sum_{q \in \mathcal N_j}\w_{jq}\Bigg)\p_{1j}\Bigg) \\
	&= \p_{ij} + \frac{1}{K}(\ell_i \otimes I_m)\bigp - \frac{1}{K}(\ell_j \otimes I_m)\bigp \\
	&= \p_{ij} + \frac{1}{K}((\ell_i - \ell_j) \otimes I_m)\bigp.
\end{aligned}
\end{equation*}
Setting $i = 1$ and ``stacking" the last line for all $j$, we obtain
\begin{equation}\label{eq:LinearDynamics}
%\begin{aligned}
%	\bigp^+ &= \bigp + \frac{1}{K}(\mathbf 1_N\ell_1 - L)\bigp \\
	\bigp^+ = \left(\left(I_N + \frac{1}{K}(\mathbf 1_N\ell_1 - L)\right) \otimes I_m\right) \bigp.
%\end{aligned}
\end{equation}
Thus, the local error dynamics are linear. It is interesting to note that the form of~\eqref{eq:LinearDynamics} implies that the dynamics on each one-parameter subgroup are decoupled. The eigenvalues of the state matrix in~\eqref{eq:LinearDynamics} are the $m$-times-repeated eigenvalues of $I + (\mathbf 1_N\ell_1 - L)/K$~\cite[Chapter 12, \S 5]{Bellman1960}.The linear dynamics~\eqref{eq:LinearDynamics} are (exponentially) stable if and only if the matrix $I + (\mathbf 1_N\ell_1 - L)/K$ is Schur. We must therefore establish conditions on the gain $K$ such that all eigenvalues of $I + (\mathbf 1_N\ell_1 - L)/K$ are in the open unit disc.

The Laplacian $L$ of the graph $\mathcal G$ is positive semidefinite, with a zero eigenvalue of algebraic multiplicity equal to the number of connected components in $\mathcal G$~\cite[Lemma 13.1.1]{GodRoy01}; the eigenvector associated with the $0$ eigenvalue is $\mathbf 1_N$.

\begin{lemma}\label{lem:Spectrum1L1}
%The spectrum of $\mathbf 1_N\ell_1 - L$ equals the spectrum of $-L$.
The spectrum of $\mathbf 1_N\ell_1 - L$ equals $\sigma(-L)$.
\end{lemma}

\begin{proof}
Let $J$ be the Jordan form of $L$ and let $V \in \Complex^{N \times N}$ be the nonsingular matrix such that $J = V^{-1}LV$, where the first column $V_1$ is in the span of $\mathbf 1_N$. We have
\begin{equation}\label{eq:L1eigs}
	V^{-1}(\mathbf 1_N\ell_1 - L)V = V^{-1}\mathbf 1_N\ell_1 V - J.
\end{equation}
Since $V_1$ is in the span of $\mathbf 1_N$ and $V^{-1}V = I$, we have $(V^{-1}\mathbf 1_N)_i = 0$ for all $i \neq 1$. Also because $V_1$ is in the span of $\mathbf 1_N$, we have $(\ell_1V)_1 = 0$. Therefore, $V^{-1}\mathbf 1_N\ell_1 V$ is strictly upper triangular. Therefore, the eigenvalues of~\eqref{eq:L1eigs} are its diagonal elements, which are the diagonal elements of $-J$, which are the negatives of the eigenvalues of $L$.
\end{proof}

\begin{lemma}\label{lem:SpectrumA}
The spectrum of $I + (\mathbf 1_N\ell_1 - L)/K$ is the image of $1 - \sigma(L)/K$.
\end{lemma}

\begin{proof}
The result follows from Lemma~\ref{lem:Spectrum1L1} and applying the Spectral Mapping Theorem~\cite[Theorem 1.13]{Higham2008} using the function $f : \Complex \to \Complex$, $f(x) = 1 - x/K$.
\end{proof}

Since the graph is assumed to be connected, $L$ has a simple eigenvalue at $0$. By Lemma~\ref{lem:SpectrumA}, this eigenvalue gets mapped to $1$ in the spectrum of $I + (\mathbf 1_N\ell_1 - L)/K$.

Let $\lambda$ be an eigenvalue of $L$ and define the function $f(x) = 1 - x/K$ as in the proof of Lemma~\ref{lem:SpectrumA}. Applying this function to $\lambda$ we have
\begin{equation*}
%\begin{aligned}
	f(\lambda) = 1 - \frac{|\lambda|}{K}\e^{j\angle\lambda} 
	= \left(1 - \frac{|\lambda|}{K}\cos(\angle\lambda)\right) -	j\frac{|\lambda|}{K}\sin(\angle\lambda)
%\end{aligned}
\end{equation*}

For stability, we require $f(\lambda)$ to be in the open unit disc. The squared magnitude of $f(\lambda)$ is
\begin{equation*}
\begin{aligned}
	|f(\lambda)|^2 &= \left(1 - \frac{|\lambda|}{K}\cos(\angle\lambda)\right)^2 + \left(\frac{|\lambda|}{K}\sin(\angle\lambda)\right)^2 \\
%	&= 1 - 2\frac{|\lambda|}{K}\cos(\angle\lambda) + \left(\frac{|\lambda|}{K}\cos(\angle\lambda)\right)%^2 \\
%		&\qquad + \left(\frac{|\lambda|}{K}\sin(\angle\lambda)\right)^2 \\
	&= \left(\frac{|\lambda|}{K}\right)^2 - 2\frac{|\lambda|}{K}\cos(\angle\lambda) + 1
\end{aligned}
\end{equation*}

Then $|f(\lambda)|^2 < 1$ if, and only if
\begin{equation*}
	\left(\frac{|\lambda|}{K}\right)^2 - 2\frac{|\lambda|}{K}\cos(\angle\lambda) < 0.
\end{equation*}

Since we have already addressed the simple eigenvalue at $0$, we assume that $\lambda \neq 0$. Therefore, dividing by $|\lambda|$ we seek conditions on $K$ such that, for all $\lambda \in \sigma{(L)} \backslash \set{0}$, $2\cos(\angle\lambda) > |\lambda|/K$. This is equivalent to the condition
\begin{equation}\label{eq:GeneralKBound}
	\left(\forall \lambda \in \sigma{(L)} \backslash \set{0} \right) \quad K > \frac{|\lambda|}{2\cos(\angle\lambda)} =\frac{|\lambda|^2}{2\mathrm{Re}(\lambda)}.
\end{equation}

If~\eqref{eq:GeneralKBound} holds, then all eigenvalues of $I + (\mathbf 1_N\ell_1 - L)/K$, except the single eigenvalue at $1$, are in the open unit disc.

The next result provides a lower bound on the controller gain $K$, as a function of the number of agents $N$, using the properties of the eigenvalues of the Laplacian of a directed graph~\cite{Agaev2005}.

\begin{lemma}\label{lem:AllBut1}
If $K > K_\mathrm{min}(N)$, where
\begin{equation}
K_\mathrm{min}(N) \coloneqq \left\{\begin{array}{rl}
	\frac{N}{2} & N \leq 9, \\
	\frac{1}{8}\csc^2(\frac{\pi}{2N})\sec(\frac{\pi}{N}) & 10 \leq N \leq 18, \\
	N - 1 & N \geq 19,
\end{array}\right.
\label{eq:K}
\end{equation}
then $I + (\mathbf 1_N\ell_1 - L)/K$ has a single eigenvalue at $1$ and all others in the open unit disc.
\end{lemma}

\begin{proof}
See Appendix.
\end{proof}

The results of~\cite{Agaev2005} allow us to find a tighter bound on $K$ than the Gershgorin Disc Theorem, which is used, for example, in~\cite{Olfati-Saber2004}.

\begin{remark}
If $\mathcal G$ is symmetric, then $\sigma(L) \subset [0,N]$. Thus~\eqref{eq:K} in Lemma~\ref{lem:AllBut1} simplifies to $K_\mathrm{min}(N) = N/2$.
\hfill $\blacktriangle$
\end{remark}

By Lemma~\ref{lem:AllBut1}, there is no $K$ for which $I + (\mathbf 1_N\ell_1 - L)/K$ is Schur. However, this does not preclude stability of~\eqref{eq:LinearDynamics}, because the
eigenvalue of $1$ corresponds to the dynamics of $\p_{11}$, the error of agent $1$ with itself, which is identically zero.

\begin{theorem}\label{thm:OneParameterGain}
Let $\G$ be a generalized cylinder. If the gain $K$ of each agent's controller~\eqref{eq:Controller} satisfies~\eqref{eq:K}, then the equilibrium $\bigp = \mathbf 0_{Nm}$ of~\eqref{eq:LinearDynamics} is locally exponentially stable
  % \max\left(\frac{N}{2},\frac{1}{4}\csc\left(\frac{\pi}{N}\right)\sqrt{\cot\left(\frac{\pi}{2N}\right)^2\sec\left(\frac{\pi}{N}\right)^2}\right)$,
  % then synchronization is achieved.
\end{theorem}

\begin{proof}
Since $E_{11}(t) = X_1^{-1}(t)X_1(t) \equiv I$, it follows immediately that $\p_{11}(t) \equiv \mathbf 0_m$. Therefore the $(N - 1)m$ dimensional subspace $\mathcal{V} \coloneqq \set{\bigp \in \Real^{Nm}: \p_{11} = \mathbf 0_m}$ is invariant under the dynamics~\eqref{eq:LinearDynamics}. As a result, we have
\begin{equation*}
\begin{aligned}
	&\sigma((I_N + (\mathbf 1_N\ell_1 - L)/K) \otimes I_m) 
	 = \sigma((I_N + (\mathbf 1_N\ell_1 - L)/K) \otimes I_m | \mathcal{V}) \sqcup \{\underbrace{1,\ldots,1}_\text{$m$ times}\},
\end{aligned}
\end{equation*}
where $(I_N + (\mathbf 1_N\ell_1 - L)/K) \otimes I_m | \mathcal{V}$ is the restriction to the subspace $\mathcal{V}$. If the gain $K$ of each agent's controller~\eqref{eq:Controller} satisfies~\eqref{eq:K}, then by Lemma~\ref{lem:AllBut1} $(I_N + (\mathbf 1_N\ell_1 - L)/K) \otimes I_m | \mathcal{V}$ is Schur.
  % The first row of $I + (\mathbf 1_N\ell_1 - L)/K$ is
  % $\begin{bmatrix}1 & 0 & \cdots & 0\end{bmatrix}$, so it is block
  % lower triangular with eigenvalue $1$ resulting from its $(1,1)$
  % element. The first row describes the dynamics of $\p_{11}$,
  % which is identically $0$. Because $\p_{11} = 0$, it does not
  % contribute to the evolution of any other error angle. Thus, the
  % first row and column of $I + (\mathbf 1_N\ell_1 - L)/K$ may
  % be deleted without losing any information about the
  % dynamics~\eqref{eq:LinearDynamics}. More formally, the
  % dynamics~\eqref{eq:LinearDynamics} evolve on the $N - 1$ dimensional
  % subspace defined by $\p_{11} = 0$. Let
  % $A \in \Real^{(N-1) \times (N-1)}$ denote the resultant matrix. The
  % spectrum of $A$ equals the spectrum of
  % $I + (\mathbf 1_N\ell_1 - L)/K$, less the eigenvalue at
  % $1$. Therefore, by Lemma~\ref{lem:AllBut1}, if
  % $K >
  % \max\left(\frac{N}{2},\frac{1}{4}\csc\left(\frac{\pi}{N}\right)\sqrt{\cot\left(\frac{\pi}{2N}\right)^2\sec\left(\frac{\pi}{N}\right)^2}\right)$,
  % then $A$ is Schur. Thus, $\p_{1j}$, $j \in \{2,\ldots,N\}$ tend
  % to $0$ exponentially fast.
\end{proof}

By Proposition~\ref{prop:1PTorus}, analogous results hold for all equilibria. We emphasize that Theorem~\ref{thm:OneParameterGain} does
not rely on Jacobian linearization of the nonlinear dynamics $E_{ij}^+$. The system in exponential coordinates evolves according to linear dynamics.

%======================================================================
\section{Performance with an Unweighted Complete Graph on a Generalized Cylinder}
%======================================================================

We define the $\varepsilon$ settling time of error $E_{ij}$ to be the smallest $\underline k \in \Int$ such that, for all $k \geq \underline k$, $E_{ij}[k] = E_{ij}[0]^\alpha$, where $|\alpha| \leq \varepsilon$, $0 < \varepsilon < 1$.

\begin{proposition}
If $\mathcal G$ is complete, then the $\varepsilon$ settling time, where $\varepsilon \in (0,1)$, is
\begin{equation*}
	T_s = \left\lceil\frac{\Log\varepsilon}{\Log\left(\frac{|K - N|}{K}\right)}\right\rceil.
\end{equation*}
\end{proposition}

\begin{proof}
\begin{align}
	E_{ij}^+ &= \left(E_{ij}\prod_{p \in \mathbb N_N \setminus \{i,j\}}E_{ip}\right)^{-\frac{1}{K}}E_{ij}\left(E_{ij}^{-1}\prod_{p \in \mathbb N_N \setminus \{i,j\}}E_{jp}\right)^{\frac{1}{K}} \nonumber \\
%	&= E_{ij}^\frac{K - 2}{K}\left(\prod_{p \in \mathbb N_N \setminus \{i,j\}}E_{ij}E_{jp}\right)^{-\frac{1}{K}}\left(\prod_{p \in \mathbb N_N \setminus \{i,j\}}E_{jp}\right)^{\frac{1}{K}} \nonumber \\
	&= E_{ij}^\frac{K - 2}{K}\left(\prod_{p \in \mathbb N_N \setminus \{i,j\}}E_{ij}^{-1}\right)^{\frac{1}{K}} \nonumber \\
	&= E_{ij}^\frac{K - 2}{K}E_{ij}^{-\frac{N - 2}{K}} \nonumber \\
	&= E_{ij}^\frac{K - N}{K} \label{eq:CompleteErrorDynamics}
\end{align}
Thus $E_{ij}[k] = E_{ij}[0]^{\left(\frac{K - N}{K}\right)^k}$.
%Applying $\phi^{-1}$ to both sides, we have
%%
%\begin{equation*}
%	\p_{ij}^+ = \frac{K - N}{K}\p_{ij}.
%\end{equation*}
%
Therefore, the $\varepsilon$ settling time is computed thus
\begin{equation*}
	\left|\frac{K - N}{K}\right|^{\underline k} = \varepsilon \implies \underline k = \frac{\Log\varepsilon}{\Log\left(\frac{|K - N|}{K}\right)}
\end{equation*}
Since $\underline k$ is a time-step, we round up to the nearest integer.
\end{proof}

The derivative of the settling time with respect to $K$ is
\begin{equation}\label{eq:dTsdn}
\begin{aligned}
	\frac{\partial T_s}{\partial K} &= \frac{\Log(\varepsilon)(|K - N|^2 + K(N - K))}{K|K - N|^2\left(\Log\left(\frac{|K - N|}{K}\right)\right)^2} \\
%	&= \frac{\Log(\varepsilon)(K^2 - 2KN + N^2 + KN - K^2)}{K|K - N|^2\left(\Log\left(\frac{|K - N|}{K}\right)\right)^2} \\
%	&= \frac{\Log(\varepsilon)(N^2 - KN)}{K|K - N|^2\left(\Log\left(\frac{|K - N|}{K}\right)\right)^2} \\
	&= \frac{\Log(\varepsilon)N(N - K)}{K|K - N|^2\left(\Log\left(\frac{|K - N|}{K}\right)\right)^2}. \\
%	&= \frac{\Log(\varepsilon)N\mathrm{sgn}(N - K)}{K|K - N|\left(\Log\left(\frac{|K - N|}{K}\right)\right)^2} \\
\end{aligned}
\end{equation}
If $K > N$, then~\eqref{eq:dTsdn} is positive, so increasing $K$, i.e., reducing the gain $1/K$, delays synchronization, which agrees with intuition. But, interestingly, if $K < N$, then~\eqref{eq:dTsdn} is negative, so increasing $K$ hastens synchronization. Although~\eqref{eq:dTsdn} is undefined at $K = N$, these observations suggest that $K = N$ is the minimizer of $T_s$.

\begin{proposition}\label{prop:Deadbeat}
If $\mathcal G$ is complete and $K = N$, then synchronization is achieved at time-step $k = 1$.
\end{proposition}

\begin{proof}
Setting $K = N$ in~\eqref{eq:CompleteErrorDynamics}, we have $E_{ij}^+ = I$.
\end{proof}

%======================================================================
\section{General Lie Groups}
\label{sec:General}
%======================================================================

For our purposes, the only difference between a generalized cylinder and any other Lie group is commutativity. Commutativity is the key property yielding Proposition~\ref{prop:Morphism} and Lemma~\ref{lem:LocalRoot}, from which all subsequent results follow. We now show that in a neighbourhood of the identity of any Lie group $\G$, commutativity holds approximately. Which has the very important implication that \emph{all} our results for generalized cylinders hold \textit{mutatis mutandis} on \emph{any} Lie group in a neighbourhood of the identity. In particular, we obtain Theorem~\ref{thm:GeneralGain}.

The Baker-Campbell-Hausdorff (BCH) formula relates the product of two elements on the Lie group $\G$ to an analytic function of their principal logarithms. If $A,B \in \g$, then the BCH formula has the series representation:
\begin{equation}\label{eq:BCH}
\begin{aligned}
	\Log(\exp(A)\exp(B)) &= A + B + \frac{1}{2}[A,B] + \frac{1}{12}[A,[A,B]] 
	 - \frac{1}{12}[B,[A,B]] + \cdots,
\end{aligned}
\end{equation}
where the remaining terms are nested brackets of increasing order~\cite[Section 3.5]{Hall2015}. We will use~\eqref{eq:BCH} to derive a linear approximation of the error dynamics on an arbitrary Lie group $\G$ near the identity, or equivalently, near the origin on the associated Lie algebra $\g$.

\begin{lemma}\label{lem:BCH}
The linearization of the BCH formula at the origin of $\g$ is $\Log(\exp(A)\exp(B)) \approx A + B$.
\end{lemma}

\begin{proof}
All nonlinear terms in~\eqref{eq:BCH} are of the form $[A,\mathrm{ad}_A^k(B)]$ and $[B,\mathrm{ad}_B^k(A)]$, $k \in \Int_{\geq 0}$. Direct computation verifies
\begin{equation*}
%\begin{aligned}
%	[A,\mathrm{ad}_A^k(B)] &= A\mathrm{ad}_A^k(B) - \mathrm{ad}_A^k(B)A \\
%	\frac{\partial}{\partial A}[A,\mathrm{ad}_A^k(B)] &= \mathrm{ad}_A^k(B) + A\frac{\partial \mathrm{ad}_A^k(B)}{\partial A} - \frac{\partial \mathrm{ad}_A^k(B)}{\partial A}A - \mathrm{ad}_A^k(B) \\
%	&= A\frac{\partial \mathrm{ad}_A^k(B)}{\partial A} - \frac{\partial \mathrm{ad}_A^k(B)}{\partial A}A \\
	\frac{\partial}{\partial A}[A,\mathrm{ad}_A^k(B)] = \left[A,\frac{\partial \mathrm{ad}_A^k(B)}{\partial A}\right].
%\end{aligned}
\end{equation*}
Thus, if~\eqref{eq:BCH} is linearized at the origin, then all nonlinear terms vanish, and $\Log(\exp(A)\exp(B)) \approx A + B$.
\end{proof}

\begin{corollary}\label{cor:exp}
Near the identity, we have $\exp(A + B) \approx \exp(A)\exp(B) \approx \exp(B)\exp(A)$.
\end{corollary}

Thus, commutativity is satisfied approximately in a neighbourhood of the identity of any Lie group $\G$. Therefore, all our results for generalized cylinders apply \textit{mutatis mutandis} to arbitrary matrix Lie groups.

\begin{corollary}\label{cor:Morphism}
Given $\p_i$ and $\p_j$ in a sufficiently small neighbourhood of zero, we have $\phi(\p_i + \p_j) \approx \phi(\p_i)\phi(\p_j)$.
\end{corollary}

\begin{corollary}\label{cor:LocalRoot}
For $\p$ sufficiently small, $\phi(\p/K) \approx \phi(t)^{1/K}$.
\end{corollary}

We state the analogues of key results for generalized cylinders. Their proofs, as well as the proof of Theorem~\ref{thm:GeneralGain}, are straightforward applications of Corollaries~\ref{cor:Morphism} and~\ref{cor:LocalRoot}.

\begin{proposition}
The equilibrium $\{E_{ij} = I : i,j \in \Nat_N\}$ is isolated.
\end{proposition}

\begin{proposition}\label{prop:GeneralStabilityIdentity}
Every equilibrium has the same stability properties as the identity.
\end{proposition}

To illustrate how the proofs of these analogues differ, we present the difference between the proofs of Propositions~\ref{prop:1PTorus} and~\ref{prop:GeneralStabilityIdentity}.

\begin{proof}
The proof differs from that of Proposition~\ref{prop:1PTorus} in only one line of the arithmetic. By Corollaries~\ref{cor:exp} and~\ref{cor:LocalRoot},
\begin{equation*}
\begin{aligned}
	\bar E_{ij}^+ &\approx \left(\prod_{p \in \mathcal N_i}\bar E_{1i}^{-\w_{ip}}\bar E_{1p}^{\w_{ip}}\right)^{-\frac{1}{K}}\bar E_{ij}\left(\prod_{q \in \mathcal N_j}\bar E_{1j}^{-\w_{jq}}\bar E_{1q}^{\w_{jq}}\right)^{\frac{1}{K}} 
	 \left(\prod_{p \in \mathcal N_i}\Xi_{1i}^{-\w_{ip}}\Xi_{1p}^{\w_{ip}}\right)^{-\frac{1}{K}}\left(\prod_{q \in \mathcal N_j}\Xi_{1j}^{-\w_{jq}}\Xi_{1q}^{\w_{jq}}\right)^{\frac{1}{K}}.
\end{aligned}
\end{equation*}

The rest of the proof is identical.
\end{proof}

%======================================================================
\section{Simulations}
\label{sec:Simulations}
%======================================================================

%----------------------------------------------------------------------
\subsection{Comparison with Kuramoto Network on
  \texorpdfstring{$\SO{2}$}{SO(2)}}
%----------------------------------------------------------------------

The Lie group $\SO{2}$ is one dimensional, thus, it is a one-parameter subgroup of $\SO{n}$ for any $n \geq 2$. $\SO{2}$ is the group of rotations in the plane, which can be interpreted locally as a position on the circumference of a circle. Given an element $R \in \SO{2}$, its local coordinate $\p \in \Real$ is often called the ``phase'' or ``angle''. The Kuramoto oscillator is a popular model of synchronization of networks of oscillators. We can view a Kuramoto network of $N$ agents as a control system, where agent $i$ has phase $\theta_i \in \Real$ with dynamics
\begin{equation}\label{eq:Kuramoto}
	\dot\theta_i = u_i, \qquad u_i = -\sum_{j \in \mathcal N_i}a_{ij}\sin(\theta_i - \theta_j),
\end{equation}
where $a_{ij} \in \Real$ is the coupling strength between agents $i$ and $j$. System~\eqref{eq:Kuramoto} can be modelled as a system on a Lie group in the form of~\eqref{eq:ctplant}, where

\begin{equation*}
	R_i = \phi(\theta) = \begin{bmatrix}\cos(\theta) & -\sin(\theta) \\ \sin(\theta) & \cos(\theta)\end{bmatrix}, \quad
	\dot R_i = R_i\begin{bmatrix}0 & -1 \\ 1 & 0\end{bmatrix}u_i.
\end{equation*}

We simulate using $N = 3$ and $a_{ij} = 1$ for all $i,j \in \Nat_N$. It can be shown that with this choice of parameters, that~\eqref{eq:Kuramoto} achieves phase synchronization~\cite{Dorfler2014}. Sampling with period $T = 0.1$, we see in Figure~\ref{fig:KuramotoT01} that synchronization is preserved under sampling. But in Figure~\ref{fig:KuramotoT08}, we see that sampling with period $T = 0.8$, that sampling destroys synchronization.

 %%%%%%%%%%%%%%%
 \begin{figure}[h!]
\centering
\includegraphics[width=0.9\textwidth]{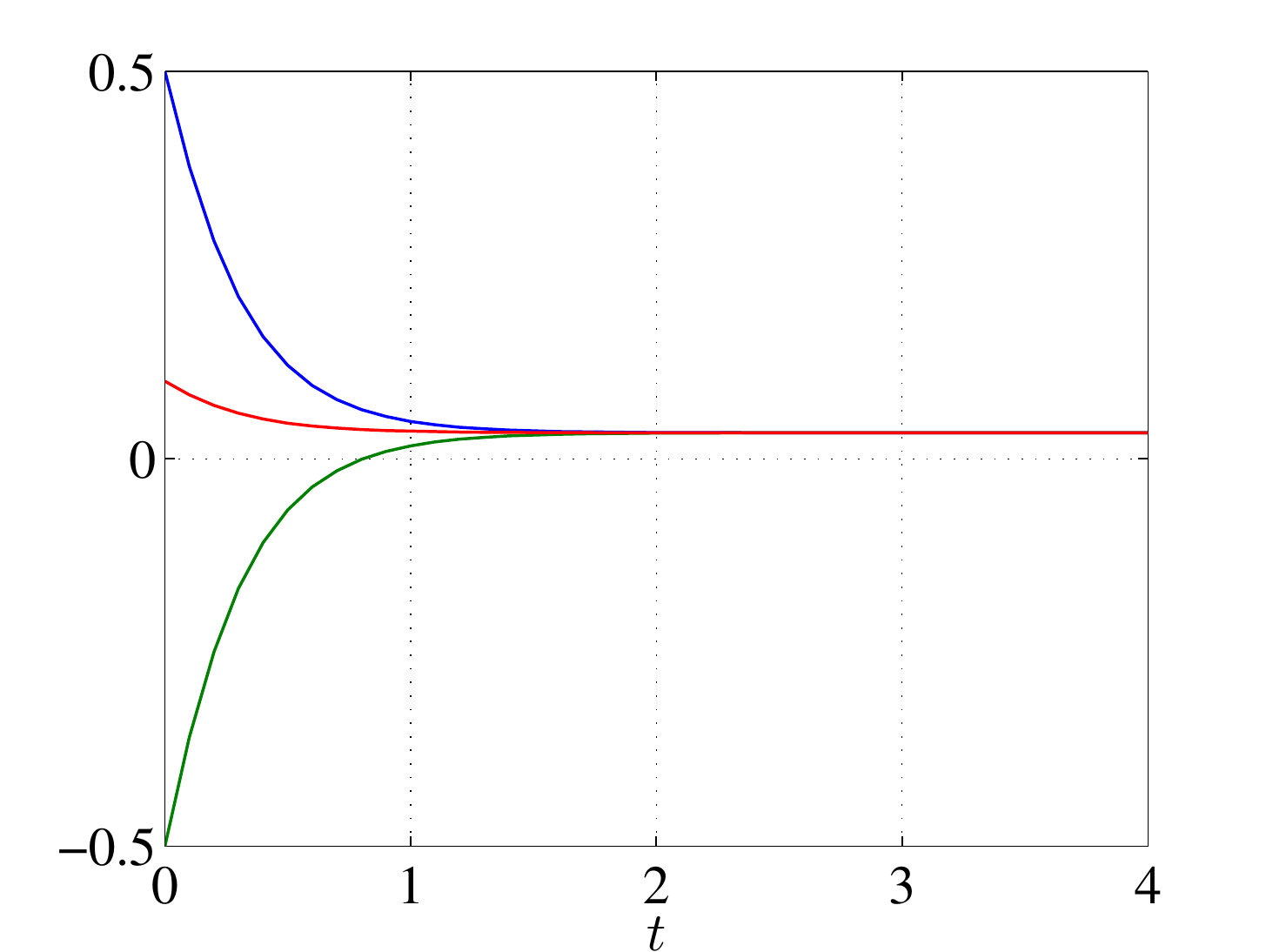}
\caption{Phases of sampled Kuramoto network with $T = 0.1$.}
\label{fig:KuramotoT01}
\end{figure}
%%%%%%%%%%%%%%%

 %%%%%%%%%%%%%%%
 \begin{figure}[h!]
\centering
\includegraphics[width=0.9\textwidth]{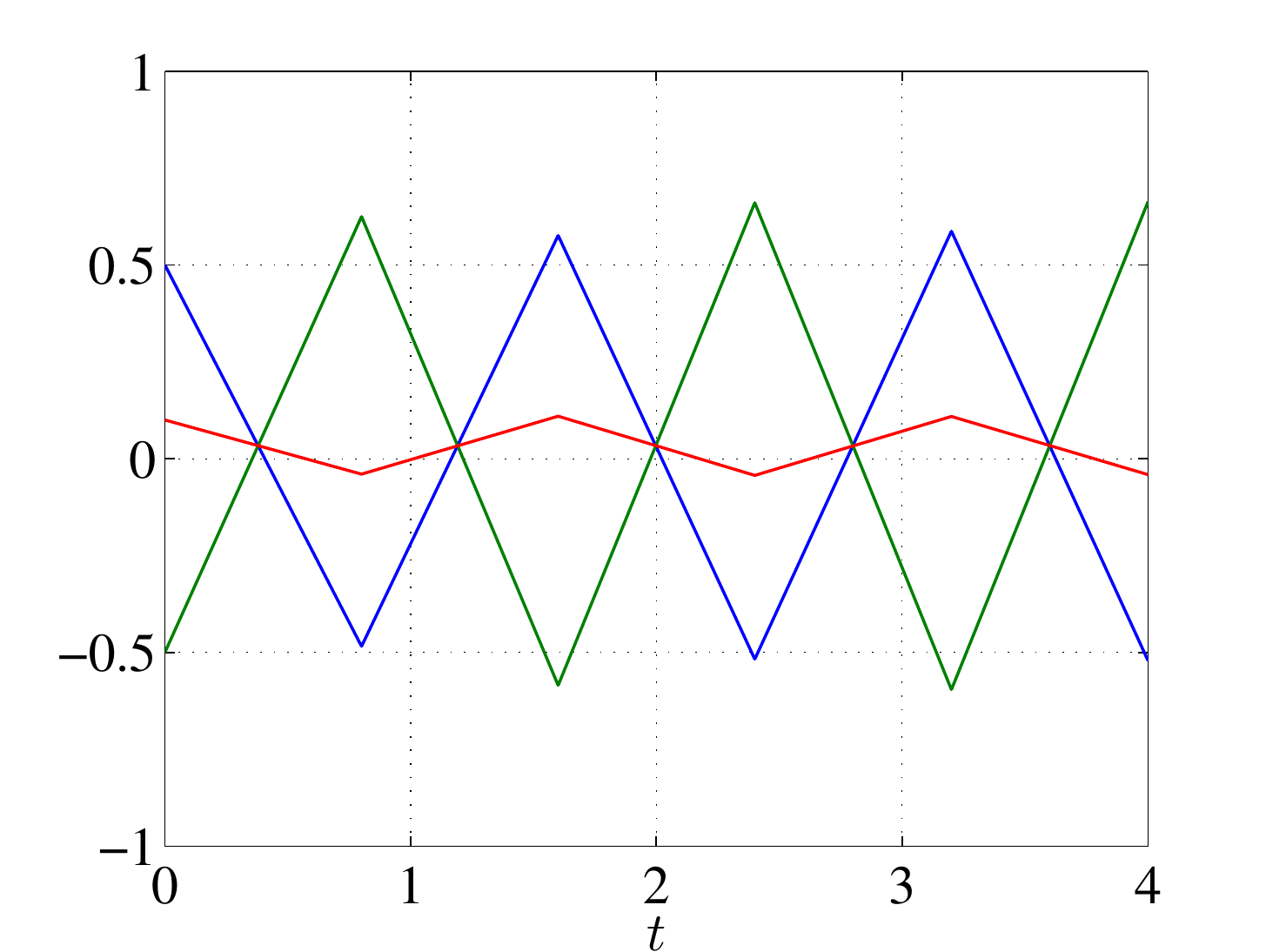}
\caption{Phases of sampled Kuramoto network with $T = 0.8$.}
\label{fig:KuramotoT08}
\end{figure}
%%%%%%%%%%%%%%%

We simulate this network again using the proposed controller with $K = 2$ and $T = 0.8$.

%%%%%%%%%%%%%%%
\begin{figure}[h!]
\centering
\includegraphics[width=0.9\textwidth]{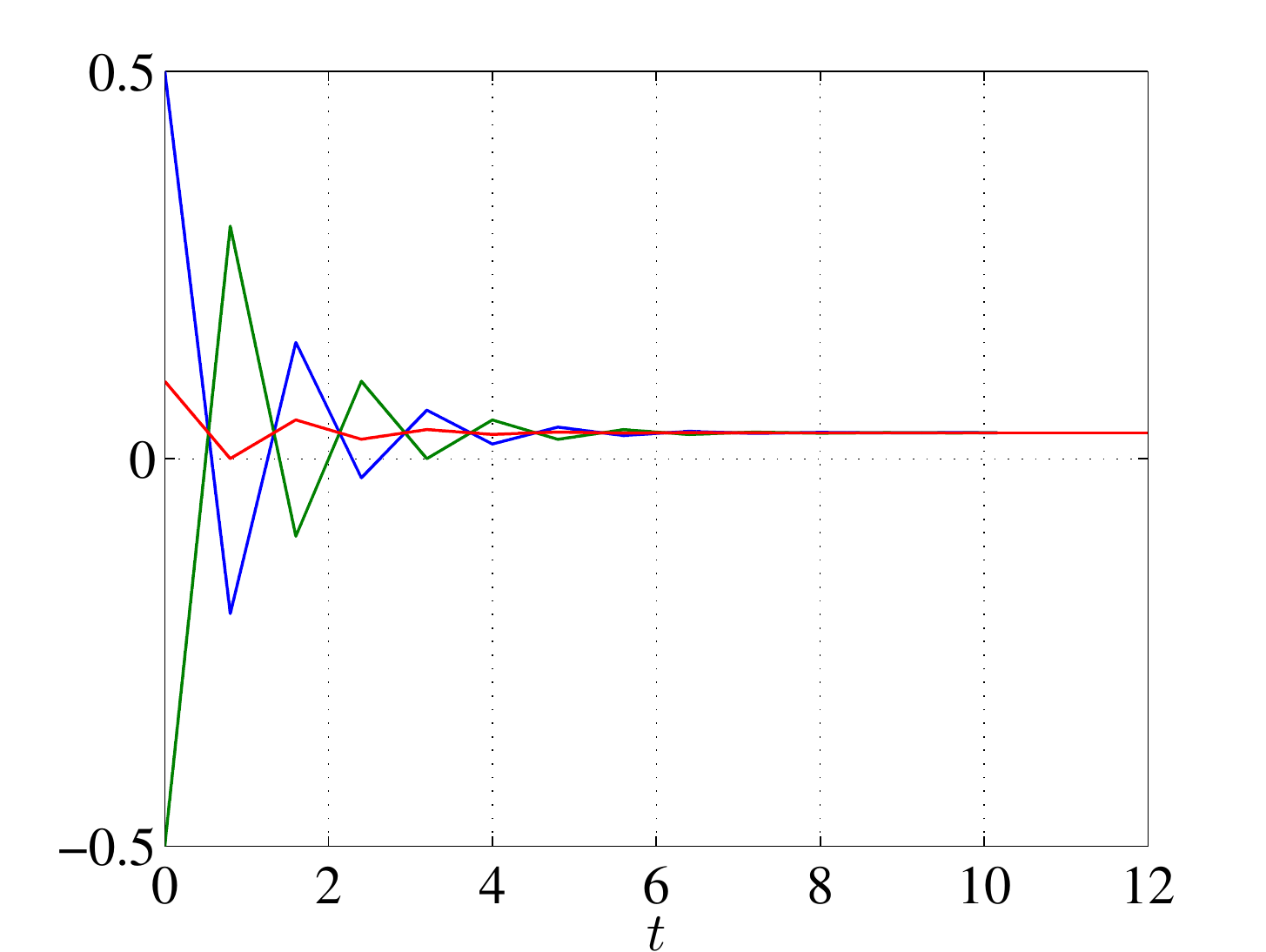}
\caption{Phases using proposed controller with $T = 0.8$.}
\label{fig:KuramotoSO2T08}
\end{figure}
%%%%%%%%%%%%%%%

As seen in Figure~\ref{fig:KuramotoSO2T08}, synchronization is achieved at $T = 0.8$, whereas it was lost using the na\"ively discretized Kuramoto coupling.

%----------------------------------------------------------------------
\subsection{Deadbeat Performance on \texorpdfstring{$\SO{2}$}{SO(2)}}
%----------------------------------------------------------------------
To illustrate Proposition~\ref{prop:Deadbeat}, we simulate a network with a complete connectivity graph on $\SO{2}$ with $K = N = 40$, $T = 1$ and initial phases evenly spaced from $-\pi/(N + 1)$ to $\pi/(N + 1)$:
\begin{equation*}
	\theta_i \coloneqq -\frac{\pi}{41} + i\frac{2\pi}{1599}.
\end{equation*}

As seen in Figure~\ref{fig:SO2FortyDeadbeat}, all error phases are driven to zero in a single time-step.

%%%%%%%%%%%%%%%
\begin{figure}[h!]
\centering
\includegraphics[width=0.9\textwidth]{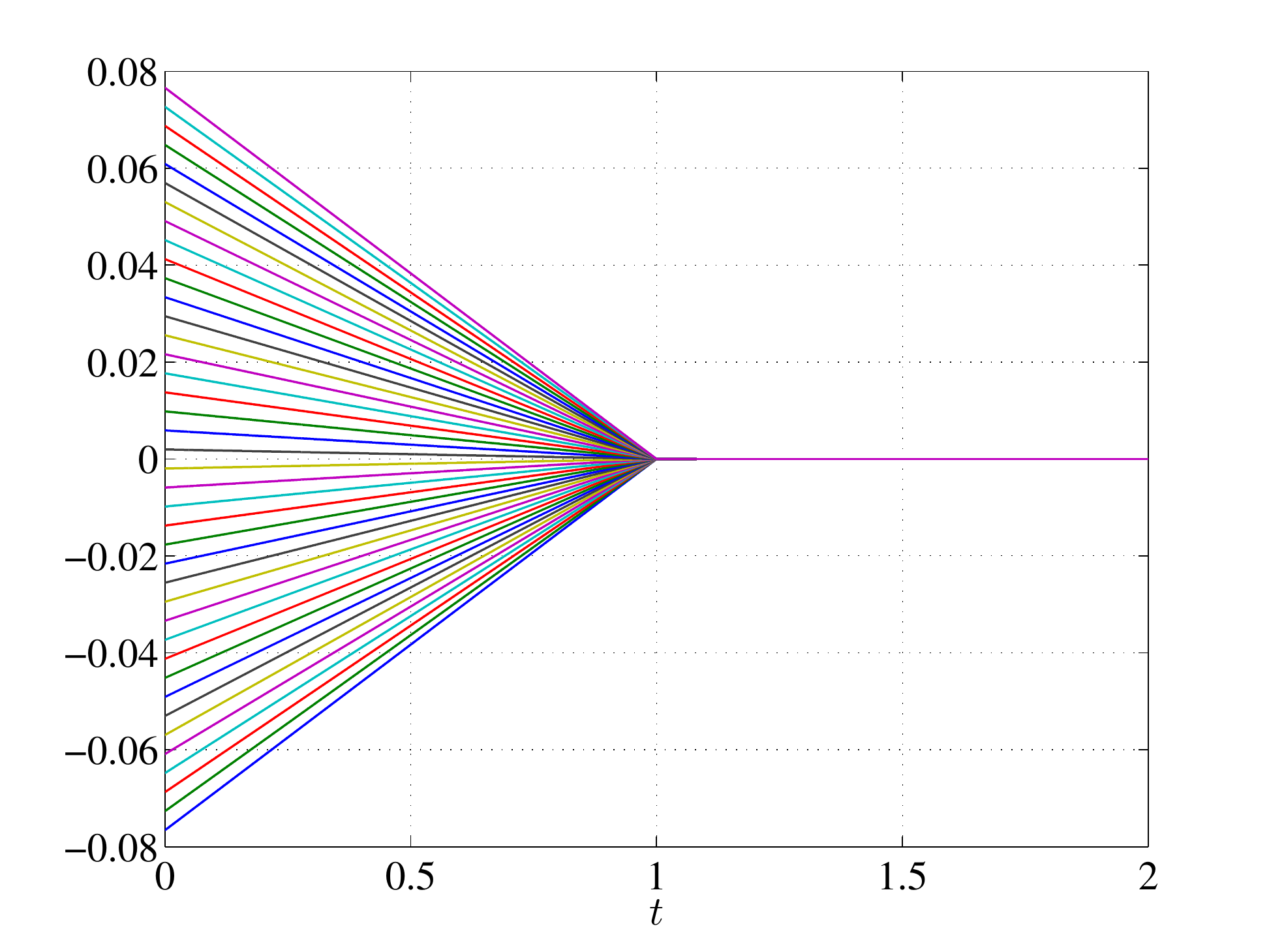}
\caption{Error phases of network on $\SO{2}$ with $T = 1$, $K = N = 40$, and $\mathcal G$ complete.}
\label{fig:SO2FortyDeadbeat}
\end{figure}
%%%%%%%%%%%%%%%

%----------------------------------------------------------------------
\subsection{Network on \texorpdfstring{$\SU{2}$}{SU(2)}}
%----------------------------------------------------------------------
We simulate a network on $\SU{2}$ to demonstrate Theorem~\ref{thm:GeneralGain} on a complex, non-commutative Lie group. We simulate a network with $N = 6$, $K = 3.5$, and graph Laplacian
%
%\begin{equation}\label{eq:Triangular}
%	L = \begin{bmatrix}
%		5 & -1 & -1 & -1 & -1 & -1 \\
%		0 & 4 & -1 & -1 & -1 & -1 \\
%		0 & 0 & 3 & -1 & -1 & -1 \\
%		0 & 0 & 0 & 2 & -1 & -1 \\
%		0 & 0 & 0 & 0 & 1 & -1 \\
%		0 & 0 & 0 & 0 & 0 & 0 \\
%	\end{bmatrix}
%\end{equation}
\begin{equation}\label{eq:Triangular}
	L = \begin{bmatrix}
		0.5 & -0.1 & -0.1 & -0.1 & -0.1 & -0.1 \\
		0 & 0.8 & -0.2 & -0.2 & -0.2 & -0.2 \\
		0 & 0 & 0.9 & -0.3 & -0.3 & -0.3 \\
		0 & 0 & 0 & 0.8 & -0.4 & -0.4 \\
		0 & 0 & 0 & 0 & 0.5 & -0.5 \\
		0 & 0 & 0 & 0 & 0 & 0 \\
	\end{bmatrix}
\end{equation}

The Pauli matrices constitute the canonical basis of $\mathfrak{su}(2)$:
\begin{equation*}
	\sigma_1 = \begin{bmatrix}0 & j \\ j & 0\end{bmatrix}, \quad
	\sigma_2 = \begin{bmatrix}0 & -1 \\ 1 & 0\end{bmatrix}, \quad
	\sigma_3 = \begin{bmatrix}j & 0 \\ 0 & -j\end{bmatrix}.
\end{equation*}
We use the Pauli matrices to generate the initial conditions:
\begin{equation*}
	U_i(0) \coloneqq \exp(a_i\sigma_1 + b_i\sigma_2 + c_i\sigma_3),
\end{equation*}
where $a_i \coloneqq -0.32 + i0.6/(N-1)$, $b_i \coloneqq -0.06 +  i0.3/(N - 1)$, $c_i \coloneqq -0.42 + i0.6/(N - 1)$.

For visualization, we plot the Euclidean norms of $\|E_{1j} - I\|$, $j \in \{2,\ldots,N\}$. As seen in Figure~\ref{fig:SU2Triangular}, the errors tend to identity, thus synchronization is achieved.

%%%%%%%%%%%%%%%
\begin{figure}[h!]
\centering
\includegraphics[width=0.9\textwidth]{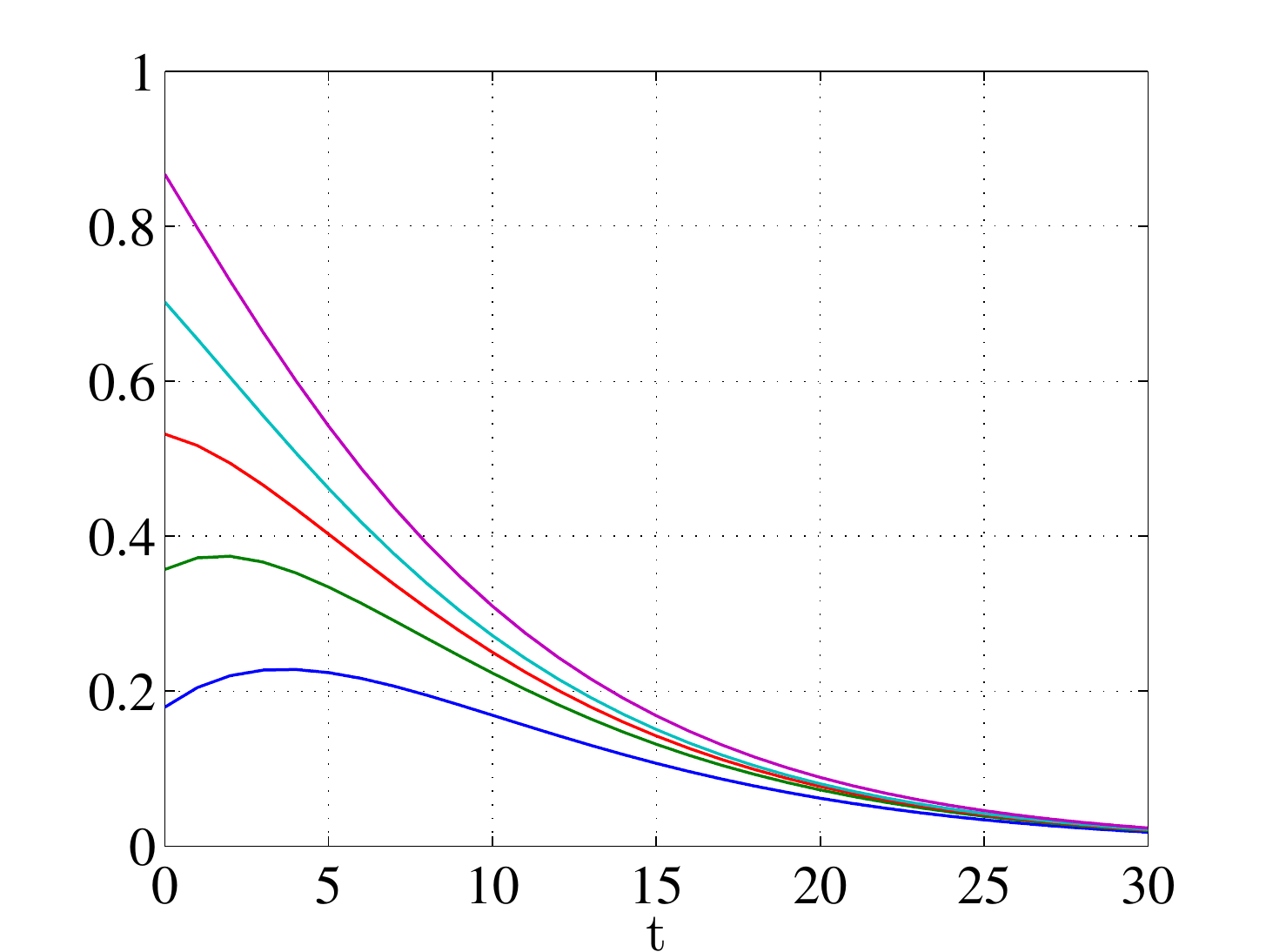}
\caption{$\|E_{1j} - I\|$, $j \in \{2,\ldots,6\}$ for a network on $\SU{2}$ with $N = 6$, $K = 3.5$, $T = 1$, and Laplacian~\eqref{eq:Triangular}.}
\label{fig:SU2Triangular}
\end{figure}
%%%%%%%%%%%%%%%

\FloatBarrier
%======================================================================
\section{Future Research}
Future work includes extending our results to agents with dynamic
models and relaxing the assumption that the agents are fully
actuated. The latter could first be addressed by assuming that the Lie
algebra generated by the input vector fields equals $\g$. It would
also be of interest to extend our results to time-varying connectivity
graphs.
%======================================================================

% B I B L I O G R A P H Y
% -----------------------

\bibliographystyle{IEEEtran}
\bibliography{arXiv}

%======================================================================

%\appendix[Proof of Lemma~\ref{lem:AllBut1}]
\section*{Appendix}
\begin{proof}[Proof of Lemma~\ref{lem:AllBut1}]
Write $\lambda \in \Complex$ in Cartesian form $\lambda \coloneqq \sigma + j\omega$, $\sigma, \omega \in \Real$. The eigenvalues of the Laplacian of a directed graph lie in the closed interior of the region in $\Complex$, whose boundary is defined the parametrized curves: $c_i(\sigma) \coloneqq \sigma + j\omega_i(\sigma)$, $i \in \Nat_5$, and their complex conjugates $\bar c_i$, $i \in \{2,3,4\}$~\cite{Agaev2005}, where the $\omega_i$ are defined by the loci:
\begin{enumerate}
\item $(\sigma - 1)^2 + \omega_1(\sigma)^2 = (N - 1)^2$,
\item $\omega_2(\sigma) = \cot\left(\frac{\pi}{N}\right)\sigma$,
\item $\omega_3(\sigma) = \frac{1}{2}\cot\left(\frac{\pi}{2N}\right)$,
\item $\omega_4(\sigma) = \cot\left(\frac{\pi}{N}\right)(N - \sigma)$,
\item $(\sigma + 1 - N)^2 + \omega_5(\sigma)^2 = (N - 1)^2$.
\end{enumerate}
If $N = 2$, then this region reduces to the interval $[0, N]$, so $K > N/2$ implies $|f(\lambda)| < 1$ for all $\lambda \in \sigma(L)\setminus\{0\}$. If $N = 3$, then this region reduces to the rhombus with vertices $0$, $N$, and $\pm j\frac{N}{2\sqrt{3}}$. For $4 \leq N \leq 18$, the region is a hexagon, as illustrated in Figure~\ref{fig:Hexagon}. If $N \geq 19$, then the region appears as in Figure~\ref{fig:Region}.

%To lower bound $K$ using only the number of agents $N$, we use the fact that the eigenvalues of a directed Laplacian lie within the parallelogram illustrated in Figure~\ref{fig:Hexagon}~\cite{Agaev2005}.
%
 %%%%%%%%%%%%%%%
 \begin{figure}[h!]
\centering
\includegraphics[width=0.9\textwidth]{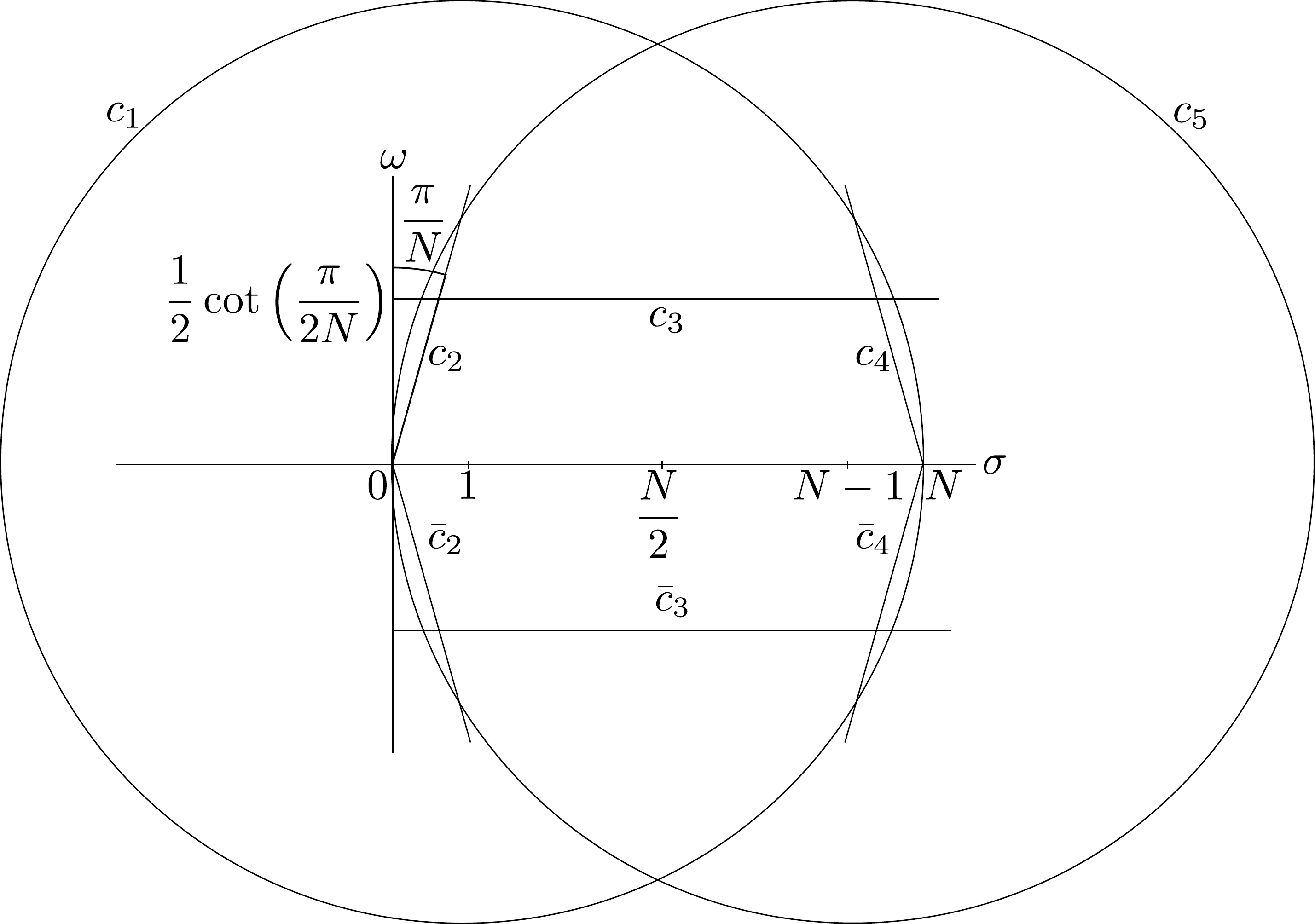}
\caption{Region containing the spectrum of the Laplacian for
  $4 \leq N \leq 18$. Figure is drawn for $N = 7$.}
\label{fig:Hexagon}
\end{figure}
%%%%%%%%%%%%%%%

 %%%%%%%%%%%%%%%
 \begin{figure}[h!]
\centering
\includegraphics[width=0.9\textwidth]{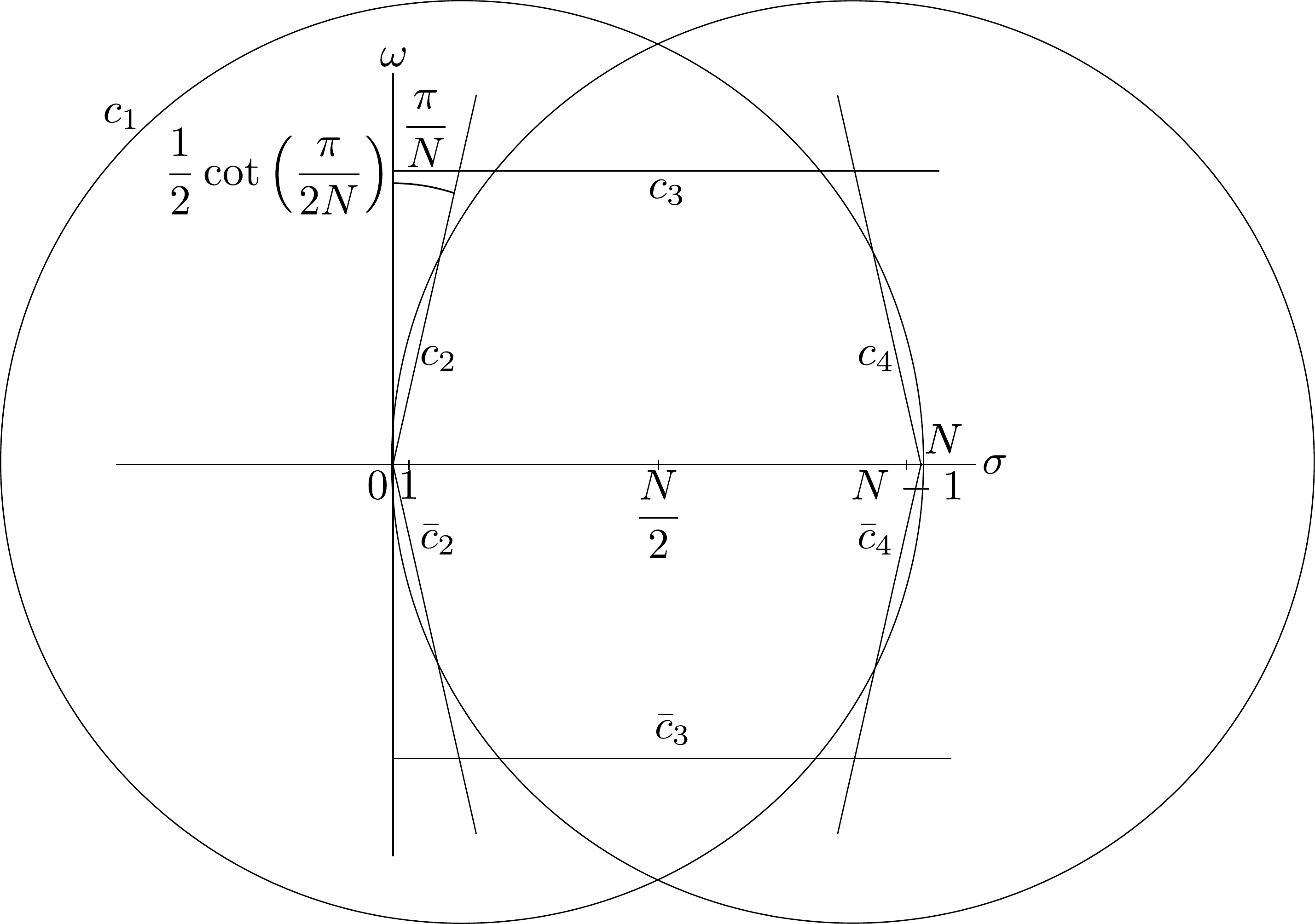}
\caption{Region containing the spectrum of the Laplacian for
  $N \geq 19$. Figure is drawn for $N = 30$.}
\label{fig:Region}
\end{figure}
%%%%%%%%%%%%%%%

To lower bound $K$ using only the number of agents $N$, we maximize the lower bound on $K$ in~\eqref{eq:GeneralKBound}, denoted by $g = 0.5|\lambda|^2/\mathrm{Re}(\lambda) = 0.5(\sigma^2 + \omega^2)/\sigma$, over this region. Since the region is a compact set, the maximum of $g$ is attained at either a critical point or at a point on the boundary of this set. The differential of $g$ is
\begin{equation*}
\begin{aligned}
\mathrm dg &= \begin{bmatrix}\frac{\partial g}{\partial \sigma} & \frac{\partial g}{\partial \omega}\end{bmatrix} \\
&= \begin{bmatrix}\frac{4\sigma^2 - 2(\sigma^2 + \omega^2)}{4\sigma^2} & \frac{\omega}{\sigma}\end{bmatrix} \\
%&= \begin{bmatrix}\frac{2\sigma^2 - \sigma^2  - \omega^2}{2\sigma^2} & \frac{\omega}{\sigma}\end{bmatrix} \\
%&= \begin{bmatrix}\frac{\sigma^2 - \omega^2}{2\sigma^2} & \frac{\omega}{\sigma}\end{bmatrix} \\
&= \begin{bmatrix}\frac{1}{2}\left(1 - \frac{\omega^2}{\sigma^2}\right) & \frac{\omega}{\sigma}\end{bmatrix},
\end{aligned}
\end{equation*}
which vanishes nowhere, thus $g$ has no critical points. Therefore, $g$ attains its maximum at a point on the boundary. We parametrize the boundary of the region by $\sigma$, and maximize $g$ on this compact, one-dimensional set. Since the Laplacian is a real matrix, its eigenvalues appear in complex conjugate pairs, so it suffices to consider the upper half complex plane.

Let $\sigma_{ij}$ denote the value of $\sigma$ at which locus $i$ intersects locus $j$. Solving $\omega_i(\sigma) = \omega_j(\sigma)$ for $\sigma$, we find:
\begin{itemize}
\item $\sigma_{35} = N - 1 - \sqrt{(N - 1)^2 - (\frac{1}{2}\cot(\frac{\pi}{2N}))^2}$,
\item $\sigma_{13} = 1 + \sqrt{(N - 1)^2 - (\frac{1}{2}\cot(\frac{\pi}{2N}))^2}$,
\item $\sigma_{14} = (N - 1)\cos(\frac{2\pi}{N}) + 1$ or $N$,
\item $\sigma_{25} = (N - 1)(1 - \cos(\frac{2\pi}{N}))$ or $0$,
%\item $\sigma_{23} = \frac{1}{1 - \tan^2(\frac{\pi}{2N})}$,
%\item $\sigma_{34} = N + \frac{1}{\tan^2(\frac{\pi}{2N}) - 1}$,
\item $\sigma_{23} = \frac{1}{2}\left(1 + \sec\left(\frac{\pi}{N}\right)\right)$,
\item $\sigma_{34} = N - \frac{1}{2}\left(1 + \sec\left(\frac{\pi}{N}\right)\right)$,
\item $\sigma_{24} = N/2$.
\end{itemize}
This boundary is illustrated for $4 \leq N \leq 18$ in Figure~\ref{fig:HexagonBoundary}, and $N \geq 19$ in Figure~\ref{fig:Boundary}.

 %%%%%%%%%%%%%%%
\begin{figure}[h!]
\centering
\includegraphics[width=0.9\textwidth]{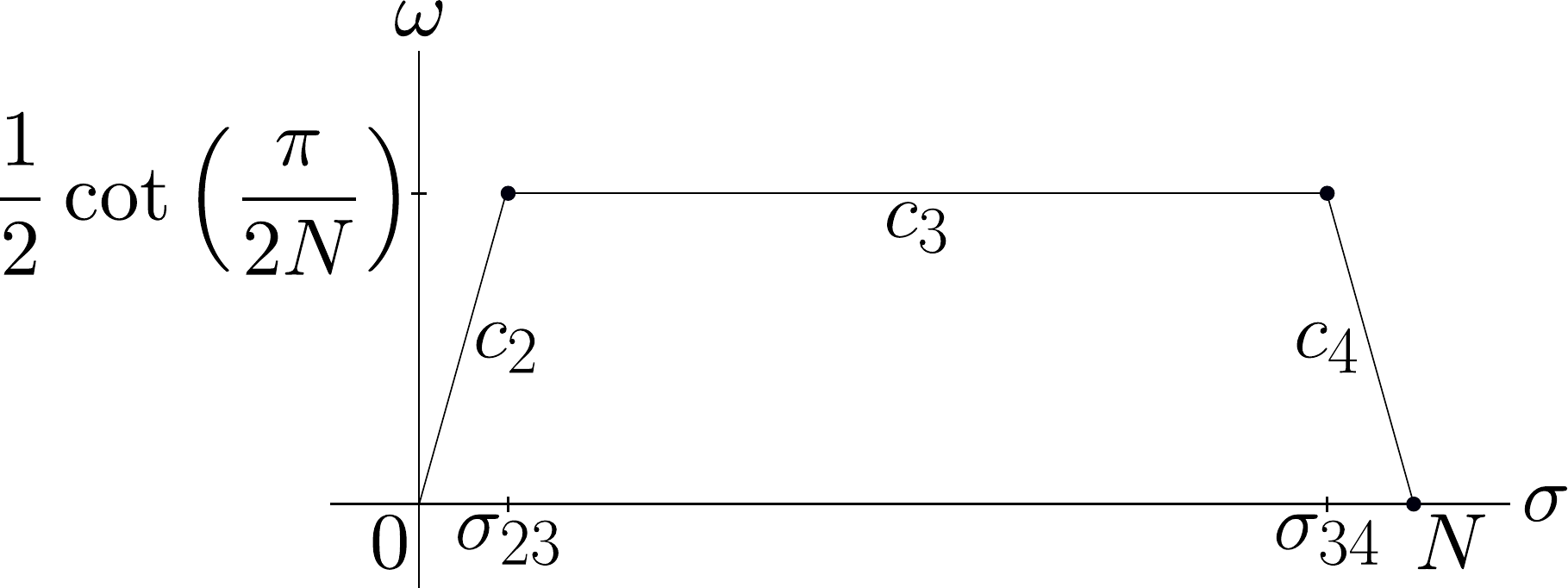}
\caption{Boundary of the region containing the eigenvalues of the Laplacian for $4 \leq N \leq 18$, $\omega \geq 0$.}
\label{fig:HexagonBoundary}
\end{figure}
%%%%%%%%%%%%%%%

 %%%%%%%%%%%%%%%
\begin{figure}[h!]
\centering
\includegraphics[width=0.9\textwidth]{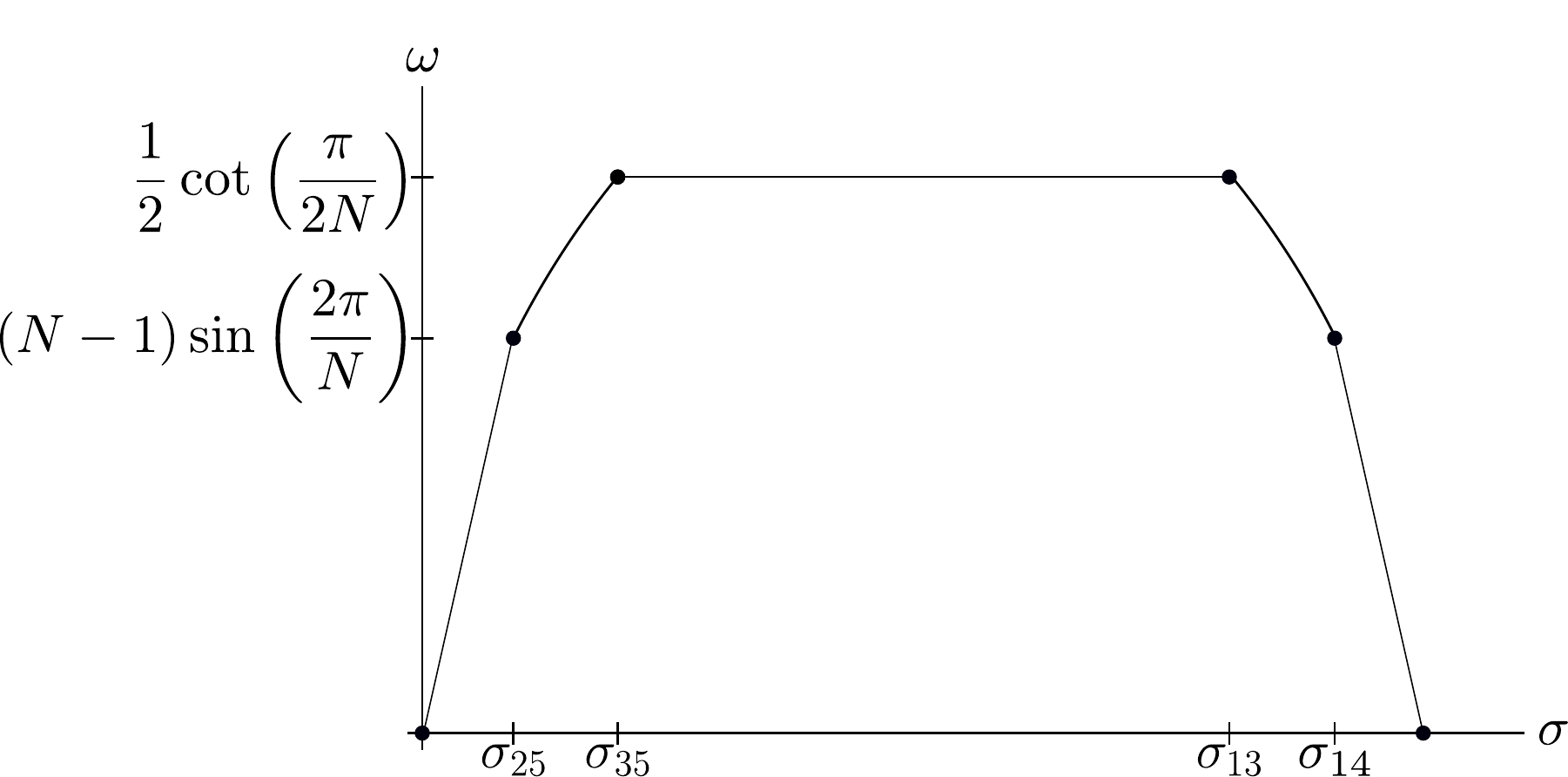}
\caption{Boundary of the region containing the eigenvalues of the Laplacian for $N \geq 19$, $\omega \geq 0$.}
\label{fig:Boundary}
\end{figure}
%%%%%%%%%%%%%%%

The Lie derivatives of $g$ in the direction of $c_i$ are:
\begin{itemize}
\item $L_{c_1}g(\sigma) = -\frac{N(N - 2)}{2\sigma^2}$,
\item $L_{c_2}g(\sigma) = \frac{1}{2}(1 + \cot^2(\frac{\pi}{N}))$,
\item $L_{c_3}g(\sigma) = \frac{1}{2} - \frac{\cot^2(\frac{\pi}{2N})}{8\sigma^2})$,
\item $L_{c_4}g(\sigma) = \frac{\csc^2\left(\frac{\pi}{N}\right)}{2} - \frac{N^2\cot^2\left(\frac{\pi}{N}\right)}{2 \sigma^2}$,
\item $L_{c_5}g(\sigma) = 0$.
\end{itemize}

Let $\sigma_i^\star$ denote the value of $\sigma$ at which $L_{c_i}g$ vanishes, which are the critical points of the restriction of $g$ to the boundary. We have
\begin{itemize}
\item $L_{c_1}g(\sigma) = 0$ if and only if $N = 0$ or $N = 2$,
\item $L_{c_2}g(\sigma) \neq 0$ for all $\sigma \in \Real$, $N \in \Nat$,
\item $\sigma_3^\star = \frac{1}{2}\cot(\frac{\pi}{2N})$,
\item $\sigma_4^\star = N\cos(\frac{\pi}{N})$,
\item $L_{c_5}g(\sigma) =0$ identically.
\end{itemize}

We determine whether $\sigma_i^\star$ is a local maximum or minimum by examining the second Lie derivative of $g$ evaluated at $\sigma_i^\star$:
\begin{equation*}
\begin{aligned}
&L^2_{c_i}g(\sigma) = \mathrm dL_{c_i}g(\sigma)\frac{\mathrm dc_i}{\mathrm d\sigma} \\
&\quad= \begin{bmatrix}\frac{\partial^2g}{\partial\sigma^2} + \frac{\partial^2g}{\partial\omega\partial\sigma}\frac{d\omega_i}{d\sigma} + \frac{\partial g}{\partial \omega}\frac{d^2\omega_i}{d\sigma^2} & \frac{\partial^2g}{\partial \sigma\partial\omega} + \frac{\partial^2g}{\partial\omega^2}\frac{d\omega_i}{d\sigma}\end{bmatrix}\begin{bmatrix}1 \\ \frac{d\omega_i}{d\sigma}\end{bmatrix} \\
&\quad= \frac{\omega^2}{\sigma^3} - 2\frac{\omega}{\sigma^2}\frac{d\omega_i}{d\sigma} + \frac{\omega}{\sigma}\frac{d^2\omega_i}{d\sigma^2} + \frac{1}{\sigma}\left(\frac{d\omega_i}{d\sigma}\right)^2.
\end{aligned}
\end{equation*}

For $i = 3$ and $i = 4$ we have
\begin{itemize}
	\item $L^2_{c_3}g(\sigma) = \cot^2(\frac{\pi}{2N})/(4\sigma^3) > 0$ for all $\sigma \geq 0$, $N \geq 2$,
	\item $L^2_{c_4}g(\sigma) = N^2\cot^2(\frac{\pi}{2})/\sigma^3 >0$ for all $\sigma \geq 0$, $N \geq 2$.
\end{itemize}
Thus, the critical points on loci $3$ and $4$ are minima, thus the maxima of $g$ on these loci restricted to the boundary are attained at the intersection points, which we now characterize.

Let $g_i \in \Real_{\geq 0}$ be the value of $g$ at $c_i(\sigma_i^\star)$, let $g_{ij} \in \Real_{\geq 0}$ be the value of $g$ at $c_i(\sigma_{ij}) = c_j(\sigma_{ij})$, and let $g_N$ be the value of $g$ at $\sigma = N$, $\omega = 0$. Since the value of $g$ is constant on locus $5$, we do not consider the intersection points of locus $5$ with loci $2$ or $3$. We have:
\begin{itemize}
\item $g_5 = N - 1$,
\item $g_{13} = 1 + \frac{N(N - 2)}{2(1 + \sqrt{(N - 1)^2 - (\frac{1}{2}\cot(\frac{\pi}{N}))^2})}$,
\item $g_{14} = 1 + \frac{N(N - 2)}{2(1 + (N - 1)\cos(\frac{2\pi}{N}))}$,
\item $g_{23} = \frac{1}{8}\csc^2(\frac{\pi}{2N})\sec(\frac{\pi}{N})$,
\item $g_{34} = \frac{\cot^2\left(\frac{\pi}{2 N}\right) + {\left(2N - 1 - \sec\left(\frac{\pi}{N}\right)\right)}^2}{4\left(2N - 1 - \sec\left(\frac{\pi}{N}\right)\right)}$,
\item $g_{24} = N(\cot^2(\frac{\pi}{N})+ 1)/4$,
\item $g_N = N/2$.
\end{itemize}

Finally, we identify the maximum value of $g$ on the boundary. For $N = 3$, the boundary is defined by only loci $2$ and $4$. $N = 3$ is also the only case in which loci $2$ and $4$ intersect. In can be shown numerically that for $N = 3$, that $g_N$ maximizes $g$. It can be verified numerically that if $4 \leq N \leq 9$, then $g_N$ is the maximum of $g$, and if $10 \leq N \leq 18$, then $g_{23}$ is the maximum of $g$, which proves the first two cases in~\eqref{eq:K}.

%%%%%%%%%%%%%%%
 \begin{figure}[h!]
\centering
\includegraphics[width=0.9\textwidth]{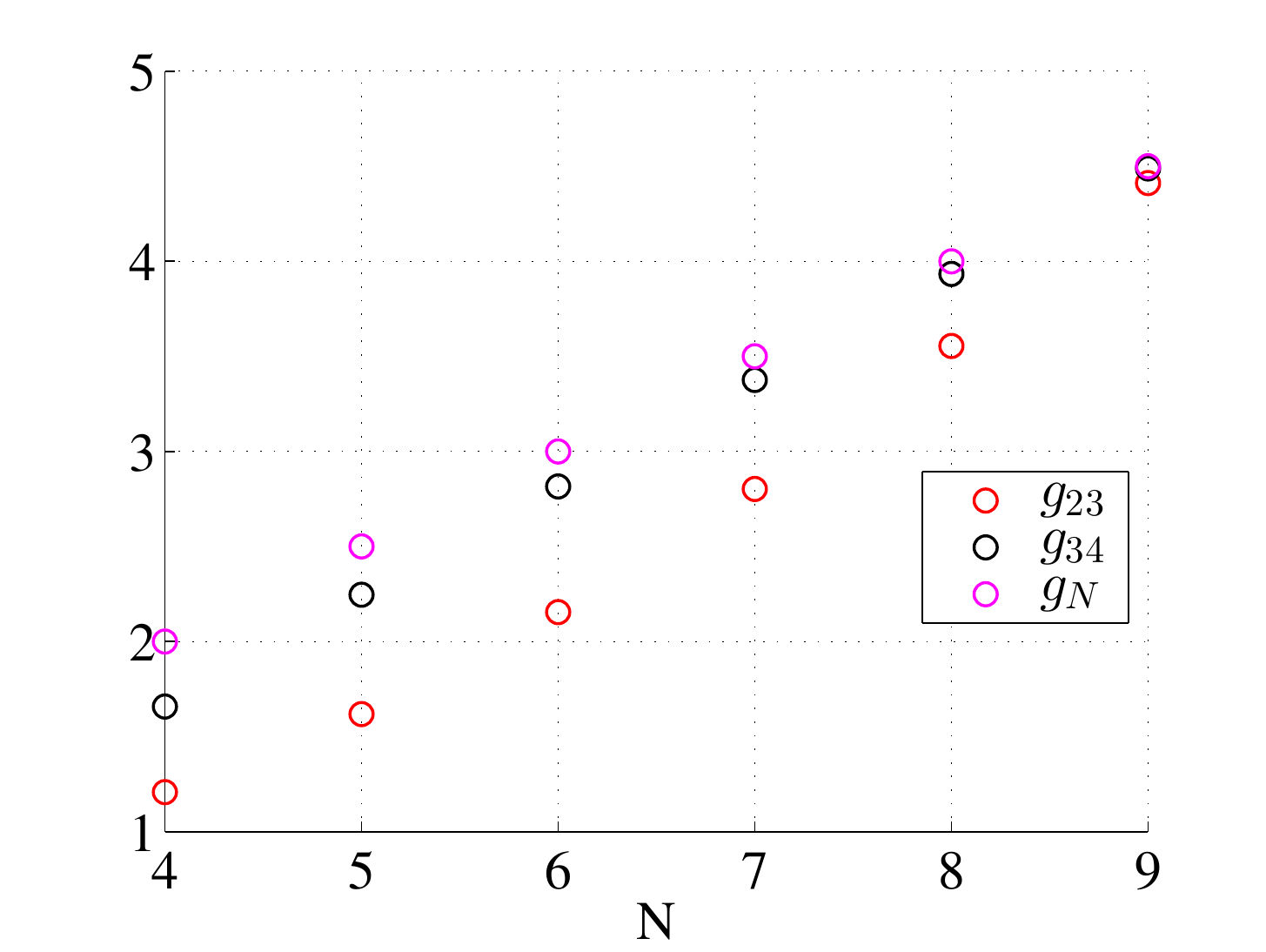}
\caption{Values of $g(N)$ for $4 \leq N \leq 9$.}
\label{fig:FourToNine}
\end{figure}
%%%%%%%%%%%%%%%

%%%%%%%%%%%%%%%
 \begin{figure}[h!]
\centering
\includegraphics[width=0.9\textwidth]{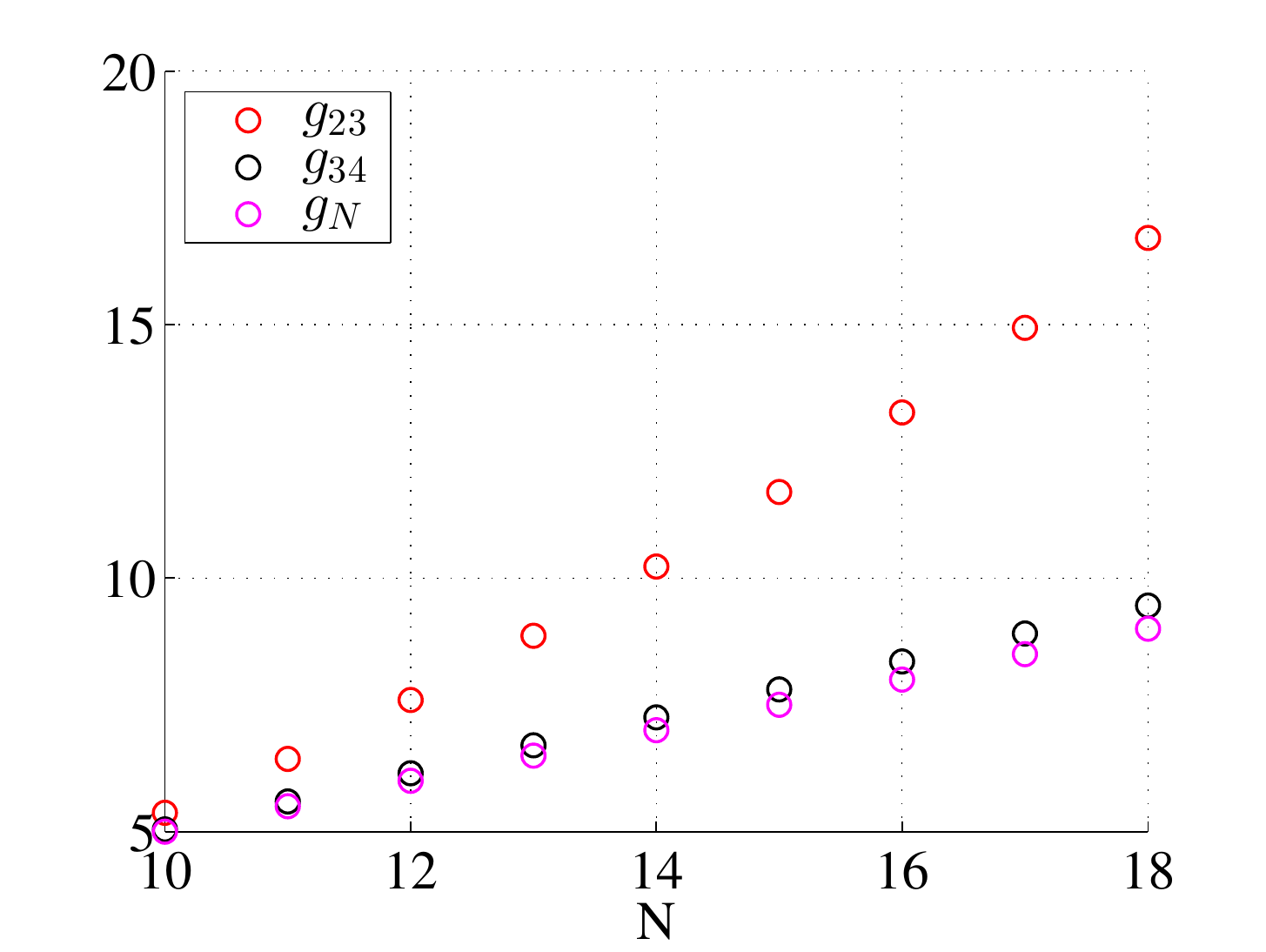}
\caption{Values of $g(N)$ for $10 \leq N \leq 18$.}
\label{fig:TenToEighteen}
\end{figure}
%%%%%%%%%%%%%%%

For $N \geq 19$, we now establish that $\mathrm{max}\{g_5, g_{13}, g_{14}\} = g_5$. Notice that $g_{13}$ and $g_{14}$ can be expressed:
\begin{equation*}
	g_{13} = 1 + \frac{N(N - 2)}{2\sigma_{13}}, \qquad g_{14} = 1 + \frac{N(N - 2)}{2\sigma_{14}}.
\end{equation*}
From their definitions, $g_5 \geq g_{13}$ if and only if
\begin{equation*}
%\begin{aligned}
	N - 1 \geq 1 + \frac{N(N - 2)}{2\sigma_{13}}
	\implies \sigma_{13} \geq \frac{N}{2}.
%\end{aligned}
\end{equation*}
Similarly, we find that $g_5 \geq g_{14}$ if and only if $\sigma_{14} > N/2$. By the geometry of the region as discussed in~\cite{Agaev2005} and illustrated in Figure~\ref{fig:Region}, these inequalities hold for all $N \geq 19$.
%That $g_5 \geq g_{14}$ if and only if $\sigma_{14} > N/2$ follows \textit{mutatis mutandis}. As seen in Figure~\ref{fig:Region}, these inequalities hold for all $N \geq 19$.
%
%The vertices to consider are $\lambda_1 \coloneqq 1 + j\frac{1}{2}\cot\left(\frac{\pi}{2N}\right)$, $\lambda_2 \coloneqq N - 1 + j\frac{1}{2}\cot\left(\frac{\pi}{2N}\right)$, and $\lambda_3 \coloneqq N$. Substituting these values into~\eqref{eq:GeneralKBound}, we obtain
%\begin{equation*}
%\begin{aligned}
%	K_1 &= \frac{1}{4}\csc\left(\frac{\pi}{N}\right)\sqrt{\cot\left(\frac{\pi}{2N}\right)^2\sec\left(\frac{\pi}{N}\right)^2} \\
%	K_2 &=\frac{1}{4} \sqrt{\frac{\left(\cot\left(\frac{\pi}{2 N}\right)^2 + \left(\sec\left(\frac{\pi}{N}\right) - 2N + 1\right)^2\right)^2}{\left(\sec\left(\frac{\pi}{N}\right) - 2N + 1\right)^2}} \\
%	K_3 &= \frac{N}{2}
%\end{aligned}
%\end{equation*}
%
%Note that in the special case of $N = 2$, the eigenvalues of $L$ are real, even if $L$ is not symmetric, thus the bounds $K_1$ and $K_2$ apply for only $N \geq 3$. Restricting $N \geq 2$, it can be shown that
%%
%\begin{equation*}
%\max\{K_1,K_2,K_3\} = 
%\left\{\begin{array}{lr}
%	K_3 & 2 \leq N < 10 \\
%	K_1 & N \geq 10.
%\end{array}\right.
%\end{equation*}
%%
%Thus proving the claim.
\end{proof}

\end{document}